\documentclass[journal, twocolumn]{IEEEtran}

\usepackage{bm}
\usepackage{cite}
\usepackage{amsmath, amssymb, amsthm, mathtools}
\usepackage{graphicx, subfigure}
\usepackage[noend]{algpseudocode}
\usepackage{algorithmicx,algorithm}
\usepackage{color}

\providecommand{\mathbold}[1]{\bm{#1}}
\newcommand{\vct}[1]{\mathbold{#1}}
\newcommand{\mtx}[1]{\mathbold{#1}}

\def \R 	{\mathbb{R}}
\def \P 	{\mathbb{P}}
\def \E 	{\mathbb{E}}
\def \mEn 	{\mtx{\mathcal{E}}}
\def \mY 	{\mtx{Y}}
\def \mD 	{\mtx{D}}

\def \mX 	{\mtx{X}}
\def \tmD 	{\tilde{\mtx{D}}}
\def \hmD 	{\hat{\mtx{D}}}
\def \mS 	{\mtx{S}}
\def \mJ 	{\mtx{J}}

\def \mA	{\mtx{A}}
\def \mZ	{\mtx{Z}}
\def \mM	{\mtx{M}}
\def \mO 	{\mtx{\Omega}}
\def \mR 	{\mtx{R}}

\def \mT 	{\mtx{T}}

\def \dist  {\textnormal{dist}}

\newtheorem{theorem}{Theorem} 
\newtheorem{lemma}{Lemma}

\newtheorem{proposition}{Proposition}

\newtheorem{definition}[lemma]{Definition}

\newtheorem{remark}{Remark}

\begin{document}
%
\title{Localization from Incomplete Euclidean Distance Matrix: Performance Analysis for the \\ SVD-MDS Approach}
%
%
\author{Huan~Zhang,
        Yulong~Liu,
	  	and~Hong~Lei}

\markboth{Draft}{Draft}
%

\maketitle

\begin{abstract}
   Localizing a cloud of points from noisy measurements of a subset of pairwise distances has applications in various areas, such as sensor network localization and reconstruction of protein conformations from NMR measurements. In \cite{Drineas2006}, Drineas \textit{et al.} proposed a natural two-stage approach, named SVD-MDS, for this purpose. This approach consists of a low-rank matrix completion algorithm, named SVD-Reconstruct, to estimate random missing distances, and the classic \textit{multidimensional scaling} (MDS) method to estimate the positions of nodes. In this paper, we present a detailed analysis for this method. More specifically, we first establish error bounds for \textit{Euclidean distance matrix} (EDM) completion in both expectation and tail forms. Utilizing these results, we then derive the error bound for the recovered positions of nodes.  In order to assess the performance of SVD-Reconstruct, we present the minimax lower bound of the zero-diagonal, symmetric, low-rank matrix completion problem by Fano's method. This result reveals that when the noise level is low, the SVD-Reconstruct approach for Euclidean distance matrix completion is suboptimal in the minimax sense; when the noise level is high, SVD-Reconstruct can achieve the optimal rate up to a constant factor.
\end{abstract}


\section{Introduction} \label{sec: introduction}

\IEEEPARstart{I}{n} many signal processing applications we work with distances because they are easy to measure. In sensor network localization, for example, each sensor simultaneously acts as a transmitter and receiver. It receives the signal sent by other sensors while emitting a signal to its surroundings. The useful information we can extract is the \textit{time-of-arrival} (TOA) or \textit{received-signal-strength} (RSS) between pairs of sensors, either of which can be seen as a metric of Euclidean distance \cite{Patwari2005}. Another example is the protein conformation problem. It has been shown by the crystallography community that after sequence-specific \textit{nuclear magnetic resonance} (NMR) assignments, we can extract the information about the intramolecular distances from two-dimensional \textit{nuclear Overhauser enhancement spectroscopy} (NOESY) \cite{Havel1985}. Other examples include geometry reconstruction of a room from echoes \cite{Dokmanic2013}, manifold learning utilizing distances \cite{Weinberger2004}, and so on.

If all the distances between pairs of nodes are available, then we can use the classic MDS algorithm \cite{Torgeson1965} to recover the coordinates of nodes. It has been proved that if all distances are measured without any error, MDS finds the configuration of nodes exactly. Moreover, MDS tolerates errors gracefully in practice, as a complete EDM overdetermines the true solution. Here, it is worth noting that we cannot recover the absolute coordinates, since rigid transformation (including rotations, translations, reflections, and their combination) does not change the EDM.

However, in many practical applications, it seems impossible to know all the entries of the EDM. In sensor network localization, for instance, due to the limit of transmission power, a sensor can only receive the signal emitted by sensors that are not too far from it. In addition, in most cases, the sensor has limited precision, which results in measurement error in distances. Thus, the measured EDM may be incomplete and noisy. In protein conformation problem, the matter is worse, because NMR spectroscopy only gives the inaccurate distances between nearby atoms. This leads to a highly incomplete EDM with noise. Therefore, it is desirable to develop methods which can localize a cloud of points from an incomplete and noisy EDM.

Generally speaking, it is a difficult task to infer missing entries of an arbitrary matrix. However, some fundamental properties of the EDM make the search for solutions feasible. It was shown in \cite{Gower1985} that the rank of an EDM is at most $d+2$, where $d$ denotes the dimension of the space in which nodes live. In other words, in most applications including sensor network localization and protein conformation problem, the rank of the EDM is at most five, as both sensors and atoms are in a three-dimensional space, although we may have thousands of sensors or atoms. By this remarkable rank property, we can solve the EDM completion problem via low-rank matrix completion approaches. Another important property of the EDM given in \cite{Gower1982} states that a necessary condition for a matrix $\mtx{D}$ is an EDM is that $-\frac{1}{2}\mtx{J}\mtx{D}\mtx{J}$ is {\em positive semidefinite} (PSD), where
\begin{equation}\label{centering_matrix}
\mtx{J} := \mtx{I}-\frac{1}{n}\vct{1}\vct{1}^T
\end{equation}
is the geometric centering matrix, $n$ denotes the number of sensors, $\vct{1}$ stands for the column vector with all ones, and $\mtx{I}$ is the identity matrix. The PSD property opens another way to solve the EDM completion problem. Most approaches in the literature utilize at least one of these two properties to localize the positions of nodes.

There have been a number of approaches proposed to determine the coordinates of nodes from incomplete EDMs in the past several decades. These methods can be roughly put into three groups based on their core ideas. The first group mainly exploits the rank property of the EDM. This group consists of algorithms that try first to estimate the missing distances by utilizing the rank property of the EDM and then use the classic MDS to find the positions from the reconstructed distance matrix. SVD-MDS \cite{Drineas2006} and OptSpace-MDS \cite{Parhizkar2013a} are two examples of this class where SVD-Reconstruct \cite{Drineas2006} and OptSpace \cite{Keshavan2010} are employed for EDM completion, respectively. The algorithms in the second group formulate the localization problem as a non-convex optimization problem and then employ different relaxation schemes to solve it. An example of this type is relaxation to a semidefinite programming \cite{Alfakih1999, Biswas2006, Javanmard2013, Ding2017}. The last group is known as metric MDS \cite{Kruskal1964, Takane1977, Parhizkar2013b}. The algorithms in this group do not try to complete the observed EDM, but directly estimate the coordinates of nodes from the incomplete EDM. Thus most efforts are paid to search for suitable cost functions and fast optimization algorithms.

Among the above mentioned methods, SVD-MDS is shown to be simple and (to some extent) effective \cite{Drineas2006}. However, theoretical understanding of this method is far from satisfactory. In this paper, we lay a solid theoretical foundation for this approach. More precisely, we establish error bounds for the recovered EDM by SVD-Reconstruct in both expectation and tail forms. Based on these results, the error bound for  recovered coordinates is derived. To show the optimality of the SVD-Reconstruct approach for EDM completion, we deduce the minimax lower bound of the zero-diagonal symmetric low-rank matrix completion problem using Fano's method \cite{Khasminskii1976}, and it reveals that when the noise level is low, SVD-Reconstruct is minimax suboptimal; when the noise level is high, SVD-Reconstruct can achieve the optimal rate up to a constant factor.

The remainder of the paper is organized as follows. In Section II, we formulate the problem and introduce the SVD-MDS approach. Section III is devoted to giving a performance analysis for the SVD-MDS method. Section IV presents the minimax lower bound of the zero-diagonal symmetric low-rank matrix completion problem. In Section V, we conclude the paper.

\section{Localization via SVD-MDS} \label{sec: localization using svd}

In this section, we formulate the localization problem and introduce the SVD-MDS approach.

\subsection{Problem Formulation} \label{subsec: problem formulation}

Given a set of $n$ nodes $\vct{x}_1, \vct{x}_2, \dots, \vct{x}_n \in \R^d$, the EDM of $\vct{x}_i$'s, denoted by $\mtx{D}$, is defined as
\begin{equation*} 
  	\mtx{D}_{ij} = \| \vct{x}_i-\vct{x}_j \|_2^2, \quad i,j=1,\dots,n.
\end{equation*}
To formalize the process of sampling the entries of $\mtx{D}$, we suppose that there is a non-zero probability $p_{ij}$ that the distance between nodes $i$ and $j$ is measured. For simplicity, we let all $p_{ij}$'s equal to a constant $p$. The observations are given as the $n \times n$ matrix $\mY$ whose entries are
\begin{displaymath}
  	\mY_{ij} = \left\{
	  	\begin{array}{cl}
		  \mD_{ij} + \mEn_{ij} & \text{with probability }p, \\
			? & \text{with probability }1-p,
		\end{array}
	\right.
\end{displaymath}
where the $?$ means that the element is unknown and $\mEn_{ij}$'s capture the effect of measurement errors, which are commonly assumed to be independent Gaussian random variables with mean zero and variance $\nu^2$. Putting this in matrix form yields
\begin{equation} \label{eq: symmetric matrix completion model}
 	\mY = \mO \odot (\mD + \mEn),
\end{equation}
where $\odot$ denotes the {\em Hadamard product} (i.e., point-wise matrix multiplication), $\mO$ is a symmetric mask matrix whose entries on or above the diagonal are independent Bernoulli random variables with parameter $p$, i.e.,
\begin{displaymath}
  	\mO_{ij} = \left\{
	  	\begin{array}{ll}
		  	1,  & \text{with probability }p \\
			0, 	& \text{with probability }1-p
		\end{array}
	\right.
	\textnormal{when }i \le j,
\end{displaymath}
and $\mEn$ is a symmetric noise matrix whose entries on or above the diagonal are independent Gaussian random variables, i.e., $\mEn_{ij} \sim N(0,\nu^2)$ when $i \le j$. The goal is to localize the cloud of points from $\mY$.

\subsection{SVD-MDS Approach} \label{subsec: svd-reconstruct}


The SVD-MDS approach is a two-stage method for localization. In the first stage, it uses SVD-Reconstruct to complete the EDM. This is done as follows. SVD-Reconstruct first constructs an unbiased estimator $\mS$ of $\mD$ with entries
\begin{displaymath}
  	\mS_{ij} = \left\{
	  	\begin{array}{cl}
		  	\frac{\mD_{ij} + \mEn_{ij} - \gamma_{ij}(1-p)}{p} & \text{with probability }p, \\
			\gamma_{ij} & \text{with probability }1-p,
		\end{array}
	\right.
\end{displaymath}
where $\gamma_{ij}$ stands for the ``best guess'' for the unknown square distances $\mD_{ij}$. Here, we always assume that $\gamma_{ij}=0$.

The next step of SVD-Reconstruct is to obtain the best rank-$r$ approximation $\tmD$ to $\mS$ ($r$ is the rank of $\mD$ and is at most $d+2$). This can be done by taking the singular value decomposition (SVD) and keeping the largest $r$ singular values and corresponding singular vectors. The original SVD-Reconstruct approach simply returns $\tmD$ as an approximation of the true $\mD$. In order to use the classic MDS for localization, we take a symmetrized version of $\tmD$ as an estimate of $\mD$, i.e.,
$$
\hmD = \frac{1}{2} \big( \tmD + \tmD^T \big).
$$

\begin{algorithm}[t]
	\caption{SVD-MDS}
  	\label{alg: SVD-Reconstruct}
	\hspace*{0.02in} {\bf Input:}
	incomplete EDM $\mY$, observation probability $p$ \\
	\hspace*{0.02in} {\bf Output:}
	coordinates $\mX$ of nodes
	\begin{algorithmic}[1]
	  \State Let $\mS = \frac{1}{p}\mY$
		\State Do SVD to $\mS$: $\mtx{U} \mtx{\Sigma} \mtx{V}^T =$ svd$(\mS)$
		\State Calculate $\tmD$ by keeping the largest $r$ sigular values of $\mS$ and the corresponding sigular vectors: $\tmD = \mtx{U}_r \mtx{\Sigma}_r \mtx{V}_r^T$
		\State Symmetrization: $\hmD = \frac{1}{2} ( \tmD + \tmD^T )$
		\State Do SVD to $-\frac{1}{2}\mJ \hmD \mJ$: $\mtx{Q} \mtx{\Lambda} \mtx{Q}^T = -\frac{1}{2}\mJ \hmD \mJ$
		\State \Return $\mX = \mtx{\Lambda}_d^{1/2}\mtx{Q}_d$
	\end{algorithmic}
\end{algorithm}

In the second stage, the classic MDS is employed to localize the nodes. The process is as follows. We first compute $-\frac{1}{2}\mJ \hmD \mJ$, where $\mtx{J}$ is defined in \eqref{centering_matrix}, and then take SVD to $-\frac{1}{2}\mJ \hmD \mJ$. Note that both $\mJ$ and $\hmD$ are symmetric, thus $-\frac{1}{2}\mJ \hmD \mJ = \mtx{Q} \mtx{\Lambda} \mtx{Q}^T$. The classic MDS simply returns $ \mtx{\Lambda}_d^{1/2}\mtx{Q}_d$ as the estimated coordinate matrix, where $\mtx{Q}_d \in \R^{n \times d}$ contains $d$ singular vectors corresponding to the $d$ largest singular values and $\mtx{\Lambda}_d$ is the $d \times d$ diagonal matrix with the $d$ largest singular values in the diagonal.

The SVD-MDS approach is summarized in Algorithm \ref{alg: SVD-Reconstruct}.

\section{Performance Analysis} \label{sec: performance analysis}

In this section, we present a detailed analysis for the SVD-MDS approach. We first establish the error bounds for EDM completion by SVD-Reconstruct in both expectation and tail forms, and then derive the error bound for coordinate recovery by MDS.

\subsection{Expectation Error Bound for EDM Completion}

In this subsection, we state the expectation version of EDM completion error via SVD-Reconstruct.
Before introducing our main result, we require the following incoherence condition \cite{Plan2014}:
\begin{equation} \label{eq: incoherence condition}
  	\mD_{ij} \le \zeta  \quad \textnormal{for any } i,j \le n,
\end{equation}
which means that each entry of $\mD$ is bounded by $\zeta$.
Our result shows that if the expected number of observed entries $m$ is large enough and the incoherence condition (\ref{eq: incoherence condition}) is satisfied, then the average error per entry by SVD-Reconstruct can be made arbitrarily smaller than $\zeta + \nu$.


\begin{theorem} [Expectation form] \label{th: expected form of matrix completion accuracy}
  Consider the model described in (\ref{eq: symmetric matrix completion model}). Let $m$ denote the expected number of observed entries, i.e., $m \coloneqq pn^2$. If $m \ge n \log n$ and the incoherence condition (\ref{eq: incoherence condition}) is satisfied, then
	\begin{equation}\label{Expectation_Form}
	  \frac{1}{n}\E \| \hat{\mtx{D}} - \mtx{D} \|_F \le C \sqrt{\frac{rn}{m}} (\zeta + \nu),
	\end{equation}
	where $C$ is an absolute constant\footnote{We use $C, C', c, c',$ and $c''$ to denote generic absolute constants, whose value may change from line to line.} and $\|\mX\|_F \coloneqq \sqrt{\textrm{trace}(\mX^T\mX)}$ denotes the Frobenius norm of $\mX$.

\end{theorem}

\begin{remark}
  Note that the left side of \eqref{Expectation_Form} measures the average error per entry of $\mtx{D}$:
  \begin{equation*}
   \frac{1}{n}\| \hat{\mtx{D}} - \mtx{D} \|_F = \left( \frac{1}{n^2}\sum_{i=1}^{n}\sum_{j=1}^{n} |\hat{\mtx{D}}_{ij} - \mtx{D}_{ij}|^2   \right)^{1/2}.
  \end{equation*}
  Thus, if $m$ is large enough, then the average error per entry can be made arbitrarily smaller than $\zeta + \nu$.

\end{remark}

\begin{remark} \label{rm: two cases for noise level}
  	When the noise level is high, i.e., $\nu \ge \zeta$, the error bound becomes
	\begin{equation*}
	  \frac{1}{n}\E \| \hat{\mtx{D}} - \mtx{D} \|_F \le C \sqrt{\frac{rn}{m}} \nu.
	\end{equation*}
	This bound, as we will see later (Theorem \ref{th: minimax lower bound for symmetric matrix completion}), is minimax optimal up to a constant factor. However, when the noise level is low, i.e., $\nu < \zeta$, we have
	\begin{equation*}
	  \frac{1}{n}\E \| \hat{\mtx{D}} - \mtx{D} \|_F \le C \sqrt{\frac{rn}{m}} \zeta,
	\end{equation*}
   which implies that the bound is minimax suboptimal in this case.
\end{remark}

\begin{remark}
  Our proof of Theorem \ref{th: expected form of matrix completion accuracy} is motivated by \cite{Plan2014}, where the authors presented a performance analysis for \eqref{eq: symmetric matrix completion model} under the assumption that both $\mO$ and $\mEn$ have independent and identical distributed (i.i.d.) entries. However, in the localization problem \eqref{eq: symmetric matrix completion model}, both $\mO$ and $\mEn$ are symmetric random matrices, which makes the analysis more difficult. Although our result (Theorem \ref{th: expected form of matrix completion accuracy}) has the similar form as that in \cite{Plan2014}, the constant $C$ may be different.
\end{remark}



The proof of Theorem \ref{th: expected form of matrix completion accuracy} makes use of some results from random matrix theory. For convenience, we include them in Appendix \ref{app: add_a}.

\begin{proof} [Proof of Theorem \ref{th: expected form of matrix completion accuracy}]
Note first that the error of $\hat{\mtx{D}}$ can be bounded by the error of $\tilde{\mtx{D}}$. Indeed,
\begin{align*}
  	\hspace{10.5pt} \E \big\| \hat{\mtx{D}}-\mtx{D} \big\|_F
	& = \E \Big\| \frac{1}{2} ( \tilde{\mtx{D}} + \tilde{\mtx{D}}^T ) -\mtx{D} \Big\|_F  \\
	& = \frac{1}{2} \E \big\| (\tilde{\mtx{D}}-\mtx{D}) + (\tilde{\mtx{D}}^T - \mtx{D}^T) \big\|_F  \\
	& \le \frac{1}{2} \Big( \E \big\| \tilde{\mtx{D}} - \mtx{D} \big\|_F + \E \big\| \tilde{\mtx{D}}^T - \mtx{D}^T \big\|_F \Big)  \\
	& \le \E \big\| \tilde{\mtx{D}} - \mtx{D} \big\|_F.
\end{align*}
The first equality comes from the definition of $\hat{\mtx{D}}$. The second equality is based on the fact that $\mtx{D}$ is a symmetric matrix. The next inequality follows from the triangle inequality. The last inequality uses the fact that the Frobenius norm of any matrix equals to the Frobenius norm of its transpose.
Since both $\tilde{\mtx{D}}$ and $\mtx{D}$ have rank $r$, $\tilde{\mtx{D}}-\mtx{D}$ has rank at most $2r$. Then we have
\begin{equation*}
  \| \tilde{\mtx{D}}-\mtx{D} \|_F \le \sqrt{2r} \| \tilde{\mtx{D}}-\mtx{D} \|,
\end{equation*}
where $\|\mX\|$ denotes the spectral norm of $\mX$, i.e., the largest singular value of $\mX$. It follows from the triangle inequality that
\begin{align*} \label{third part in proof of expectation}
  	\| \tilde{\mtx{D}}-\mtx{D} \| & \le \| \tilde{\mtx{D}}-p^{-1} \mtx{\Omega} \odot (\mtx{D} + \mEn)\| + \| p^{-1} \mtx{\Omega} \odot (\mtx{D} + \mEn)- \mtx{D} \|  \\
	& \le 2p^{-1}\| \mtx{\Omega} \odot (\mtx{D} + \mEn) - p\mtx{D} \|  \\
	& \le 2p^{-1}\| \mtx{\Omega} \odot \mtx{D} - p\mtx{D} \| + 2p^{-1}\| \mO \odot \mEn\|.
\end{align*}
The second inequality holds because $\tilde{\mtx{D}}$ is the best rank-$r$ approximation to $p^{-1} \mtx{\Omega} \odot (\mtx{D}+\mEn)$. Therefore, it suffices to bound $\E \| \mtx{\Omega} \odot \mtx{D} - p\mtx{D} \|$ and $\E\| \mO \odot \mEn\|$.


To bound $\E \| \mtx{\Omega} \odot \mtx{D} - p\mtx{D} \|$, let $\mtx{\Omega} = \mtx{\Omega}_u + \mtx{\Omega}_l$, where $\mtx{\Omega}_u$ is the upper-triangle matrix containing the entries of $\mtx{\Omega}$ on or above the diagonal and $\mtx{\Omega}_l$ is the lower-triangle matrix containing the entries of $\mtx{\Omega}$ below the diagonal. Let $\mtx{\Omega}'$ be a random matrix (independent of $\mO$ and $\mEn$) with i.i.d. entries satisfying the following distribution: $\P \{ \mtx{\Omega}'_{ij} = 1 \} = p$ and $\P \{ \mtx{\Omega}'_{ij} = 0 \} = 1-p$. Similar to $\mtx{\Omega}_u$ and $\mtx{\Omega}_l$, we also define $\mtx{\Omega}'_u$, $\mtx{\Omega}'_l, \mtx{D}_u$, and $\mtx{D}_l$. Then, $\E \big\|\mtx{\Omega} \odot \mtx{D} - p\mtx{D} \big\|$ can be bounded as follows:
\begin{align*}
 	&\E \big\| \mtx{\Omega} \odot \mtx{D} - p\mtx{D} \big\| \\
	& \hspace{4em}= \E \big\| (\mtx{\Omega}_u \odot \mtx{D} - p\mtx{D}_u) + (\mtx{\Omega}_l \odot \mtx{D} -p\mtx{D}_l) \big\|  \\
	& \hspace{4em}\le \E \big\| \mtx{\Omega}_u \odot \mtx{D} - p\mtx{D}_u \big\| + \E \big\| \mtx{\Omega}_l \odot \mtx{D} -p\mtx{D}_l \big\|  \\
	& \hspace{4em}= \E \big\| \mtx{\Omega}'_u \odot \mtx{D} - p\mtx{D}_u \big\| + \E \big\| \mtx{\Omega}'_l \odot \mtx{D} -p\mtx{D}_l \big\|  \\
	& \hspace{4em}= \E \big\| (\mtx{\Omega}'_u \odot \mtx{D} - p\mtx{D}_u) + \E(\mtx{\Omega}'_l \odot \mtx{D} - p\mtx{D}_l) \big\| + {} \\
	& \hspace{4em}\hspace{13.6pt} \E \big\| (\mtx{\Omega}'_l \odot \mtx{D} -p\mtx{D}_l) + \E(\mtx{\Omega}'_u \odot \mtx{D} - p\mtx{D}_u) \big\| \\
	& \hspace{4em}\le \E \big\| (\mtx{\Omega}'_u \odot \mtx{D} - p\mtx{D}_u) + (\mtx{\Omega}'_l \odot \mtx{D} - p\mtx{D}_l) \big\| + {} \\
	& \hspace{4em}\hspace{13.6pt} \E \big\| (\mtx{\Omega}'_l \odot \mtx{D} -p\mtx{D}_l) + (\mtx{\Omega}'_u \odot \mtx{D} - p\mtx{D}_u) \big\| \\
	& \hspace{4em}= 2\E \big\| \mtx{\Omega}' \odot \mtx{D} - p\mtx{D} \big\|.
\end{align*}
 The first inequality is a consequence of the triangle inequality. The second equality holds because $\mtx{\Omega}_u$ and $\mtx{\Omega}_l$ share the same distribution with $\mtx{\Omega}'_u$ and $\mtx{\Omega}'_l$, respectively, so the expectations are equal. The second inequality results from Jensen's inequality. Since now the mask matrix $\mtx{\Omega}'$ has i.i.d. random entries, we can bound $\E \big\| \mtx{\Omega}' \odot \mtx{D} - p\mtx{D} \big\|$ using standard tools from random matrix theory. We proceed by using symmetrization technique (Lemma \ref{th: symmetrization}) first and then contraction principle (Lemma \ref{th: contraction principle}):
\begin{equation*}
  \E \big\| \mtx{\Omega}' \odot \mtx{D} - p\mtx{D} \big\| \le 2 \E \| \mtx{R} \odot \mtx{\Omega}' \odot \mtx{D} \big\| \le 2 \zeta \E \big\| \mtx{R} \odot \mtx{\Omega}' \big\|,
\end{equation*}
where $\mtx{R}$ is the matrix whose entries are independent symmetric Bernoulli random variables (i.e., $\mtx{R}_{ij}=1$ or $\mtx{R}_{ij}=-1$ with probability 1/2). It is clear that the random matrix $\mtx{R} \odot \mtx{\Omega}'$ has i.i.d. entries. Therefore, $\E \big\| \mtx{R} \odot \mtx{\Omega}' \big\|$ can be bounded by Seginer's theorem (Lemma \ref{th: Seginer's theorem}):
$$
\E \big\| \mtx{R} \odot \mtx{\Omega}' \big\| \le C_1 \Big( \E \max_i \|(\mtx{R} \odot \mtx{\Omega}')_{i \cdot }\|_2 + \E \max_j \|(\mtx{R} \odot \mtx{\Omega}')_{\cdot j}\|_2 \Big),
$$
where $(\mtx{R} \odot \mtx{\Omega}')_{i \cdot}$ and $(\mtx{R} \odot \mtx{\Omega}')_{\cdot j}$ denote the $i$-th row and the $j$-th column of $\mtx{R} \odot \mtx{\Omega}'$ respectively and $C_1$ is an absolute constant.

It is not hard to see that $\|(\mtx{R} \odot \mtx{\Omega}')_{i \cdot}\|_2^2$ follows the binomial distribution with parameter $(n,p)$. By Jensen's inequality and Lemma \ref{th: Maximum of binomials}, we have
$$
\E \max_{1 \le i \le n} \|(\mR \odot \mO)'_{i \cdot}\|_2 \le \big( \E \max_{1 \le i \le n} \|(\mR \odot \mO)'_{i \cdot}\|_2^2 \big)^{1/2} \le c\sqrt{np}.
$$
Similarly,
$$
\E \max_{1 \le j \le n} \|(\mR \odot \mO)'_{\cdot j}\|_2 \le c\sqrt{np}.
$$
Therefore,
$$
  \E \big\| \mtx{R} \odot \mtx{\Omega}' \big\| \le C \sqrt{np},
$$
where $C$ is an absolute constant. Thus, the first term can be bounded by
$$
\E \| \mO \odot \mD - p\mtx{D} \| \le C\zeta \sqrt{np},
$$
where $C$ is a constant number.

The second term $\E\| \mO \odot \mEn\|$ can be bounded similarly. Define $\mEn'$, $\mEn'_u$, $\mEn'_l$, $\mEn_u$, and $\mEn_l$ in the same way as $\mO'$, $\mO'_u$, $\mO'_l$, $\mO_u$, and $\mO_l$. 
Then we have
$$	
\E \| \mtx{\Omega} \odot \mEn \| \le \E \| \mO'_u \odot \mEn'_u \| + \E \| \mO'_l \odot \mEn'_l \| \le 2\E \| \mO' \odot \mEn' \|.
$$
The first inequality comes from the triangle inequality, and the second from Jensen's inequality. Since $\mO' \odot \mEn'$ has i.i.d. entries, Seginer's theorem gives
$$
\E \big\| \mO' \odot \mEn' \big\| \le C_1 \Big( \E \max_i \|(\mO' \odot \mEn')_{i \cdot }\|_2 + \E \max_j \|(\mO' \odot \mEn')_{\cdot j}\|_2 \Big).
$$
Here, $\E \max_i \|(\mO' \odot \mEn')_{i \cdot }\|_2$ can be bounded by the second part of Lemma \ref{th: Maximum of binomials}:
$$
\E \max_{i} \| (\mO' \odot \mEn')_{i \cdot} \|_2 \le C'\sqrt{np}.
$$
The same argument holds for $\E \max_{j} \| (\mO' \odot \mEn')_{\cdot j} \|_2$. Thus, we conclude that
$$
\E \| \mO \odot \mEn \| \le C'\sqrt{np}.
$$
Putting all these together, we obtain the desired result:
$$
\frac{1}{n}\E \| \hat{\mtx{D}} - \mtx{D} \|_F \le C \sqrt{\frac{rn}{m}} (\zeta + \nu).
$$
\end{proof}

\subsection{Tail Bound for EDM Completion} \label{subsec: tail bound}

In this subsection, we derive the tail bound for the EDM completion error, which shows that the error probability decreases fast as the error increases.
\begin{theorem} [Tail Form] \label{th: tail bound for matrix completion error}
  Consider the model described in (\ref{eq: symmetric matrix completion model}). If the incoherence condition (\ref{eq: incoherence condition}) is satisfied, then for any $t \ge 0$,
	\begin{multline} \label{eq: tail bound}
	  	\P \Big\{ \| \hat{\mtx{D}}-\mtx{D} \|_F \ge t \Big\} \le \\
	  	n \exp \Big[ -c \cdot \min \Big( \frac{mt^2}{n^3r(\zeta+\nu)^2}, \frac{mt}{n^2\sqrt{r}(\zeta+\nu)} \Big) \Big],
	\end{multline}
	where $c$ is an absolute constant.
\end{theorem}



\begin{proof}
  	We first proceed similarly as in the proof of Theorem \ref{th: expected form of matrix completion accuracy}:
	\begin{align*}
		\big\| \hat{\mtx{D}}-\mtx{D} \big\|_F & \le \big\| \tilde{\mtx{D}}-\mtx{D} \big\|_F \le \sqrt{2r} \big\| \tilde{\mtx{D}}-\mtx{D} \big\| \\
		& \le 2\sqrt{2r}p^{-1} \big\| \mtx{\Omega} \odot (\mtx{D}+\mEn) - p\mtx{D} \big\|.
  	\end{align*}
	So it suffices to bound the tail of $\big\| \mtx{\Omega} \odot (\mtx{D}+\mEn) - p\mtx{D} \big\|$. Let $\mtx{A}:=\mtx{\Omega} \odot (\mtx{D}+\mEn) - p\mtx{D}$, and let $\mtx{A}_{ij}$ denote the entry of $\mtx{A}$ on the $i$-th row and $j$-th column. We define the following $\frac{1}{2}n(n+1)$ matrices  using $\mtx{A}_{ij}$: When $i < j$,
	$$
	\mtx{Z}_{ij}= \mtx{A}_{ij}(\vct{e}_i\vct{e}_j^T + \vct{e}_j\vct{e}_i^T),
	$$
	and when $i=j$,
	$$
	\mtx{Z}_{ii} = \mtx{A}_{ij} \vct{e}_i\vct{e}_i^T,
	$$
	where $\vct{e}_i$ is the standard basis vector with a one in position $i$ and zeros elsewhere.
	Clearly, $\{\mtx{Z}_{ij}\}$ is a sequence of centered, independent, self-adjoint random matrices. Moreover, for $k = 2,3\dots$, when $i < j$, we have
	\begin{align*}
		\E \mtx{Z}_{ij}^k
		&= \E \mtx{A}_{ij}^k (\vct{e}_i\vct{e}_j^T + \vct{e}_j\vct{e}_i^T)^k \\
		&\preceq |\E \mtx{A}_{ij}^k| \cdot (2\vct{e}_i\vct{e}_i^T + 2\vct{e}_j\vct{e}_j^T + \vct{e}_i\vct{e}_j^T + \vct{e}_j\vct{e}_i^T) \\
		&\preceq \E \big[ p \cdot \big|(1-p)\mD_{ij}+\mEn_{ij}\big|^k+(1-p)\cdot \big|-p\mD_{ij}\big|^k\big] \\
	  	& \hspace{1em} \cdot (2\vct{e}_i\vct{e}_i^T + 2\vct{e}_j\vct{e}_j^T + \vct{e}_i\vct{e}_j^T + \vct{e}_j\vct{e}_i^T) \\
		&\preceq p \cdot \E \big[ \big|(1-p)\mD_{ij}+\mEn_{ij}\big|^k+\big|\mD_{ij}\big|^k\big] \\
	  	& \hspace{1em} \cdot (2\vct{e}_i\vct{e}_i^T + 2\vct{e}_j\vct{e}_j^T + \vct{e}_i\vct{e}_j^T + \vct{e}_j\vct{e}_i^T),
  	\end{align*}
	where $\mA \preceq \mtx{B}$ means that $\mtx{B} - \mtx{A}$ is positive semidefinite. The first inequality holds because $(\vct{e}_i\vct{e}_j^T + \vct{e}_j\vct{e}_i^T)^k$ has periodicity for $k \ge 2$, and we can verify the inequality by a direct calculation. The second inequality results from Jensen's inequality and the fact that the matrix $2\vct{e}_i\vct{e}_i^T + 2\vct{e}_j\vct{e}_j^T + \vct{e}_i\vct{e}_j^T + \vct{e}_j\vct{e}_i^T$ is positive semidefinite. The third inequality again uses the fact that the matrix $2\vct{e}_i\vct{e}_i^T + 2\vct{e}_j\vct{e}_j^T + \vct{e}_i\vct{e}_j^T + \vct{e}_j\vct{e}_i^T$ is positive semidefinite. Note that $\mEn_{ij}$ is a normal variable $\mEn_{ij} \sim N(0, \nu^2)$, and hence it is sub-exponential, with its $\psi_1$ norm bounded by
	$
	\|\mEn_{ij}\|_{\psi_1} \le c_1\nu.
	$
	Then the $\psi_1$ norm of $(1-p)\mD_{ij}+\mEn_{ij}$ can be bounded by the triangle inequality:
	$$
	\|(1-p) \mD_{ij} + \mEn_{ij}\|_{\psi_1} \le \|(1-p)\mD_{ij}\|_{\psi_1} + \|\mEn_{ij}\|_{\psi_1} \le c_2(\zeta + \nu).
	$$
	Now, we can calculate the moment of $(1-p) \mD_{ij} + \mEn_{ij}$ using the Integral identity (Lemma \ref{th: Integral identity}) and the tail property of sub-exponential random variables:
	\begin{align*}
	  &\hspace{-1.5em}\E |(1-p) \mD_{ij} + \mEn_{ij}|^k \\
	  &\hspace{1.5em}= \int_0^{\infty} \P\big\{ |(1-p) \mD_{ij} + \mEn_{ij}|^k \ge u\big\} du \\
		&\hspace{1.5em}= \int_0^{\infty} \P\big\{ |(1-p) \mD_{ij} + \mEn_{ij}| \ge t\big\} kt^{k-1}dt \\
		&\hspace{1.5em}\le \int_0^{\infty} \exp(1-t/K_1) kt^{k-1}dt \\
		&\hspace{1.5em}\le \int_0^{\infty} \exp(1-s) k K_1^k s^{k-1}ds \\
		&\hspace{1.5em}= ekK_1^k \Gamma(k) \\
		&\hspace{1.5em}= eK_1^k k!,
  	\end{align*}
	where $K_1 = c_3(\zeta + \nu)$ and $c_3 > 0$ is an absolute constant. It follows that when $i < j$ we have
	\begin{align*}
		\E \mtx{Z}_{ij}^k &\preceq p(eK_1^k k! +\zeta^k)(2\vct{e}_i\vct{e}_i^T + 2\vct{e}_j\vct{e}_j^T + \vct{e}_i\vct{e}_j^T + \vct{e}_j\vct{e}_i^T) \\
		&\preceq ep\big[c_4(\zeta+\nu)\big]^k k! (2\vct{e}_i\vct{e}_i^T + 2\vct{e}_j\vct{e}_j^T + \vct{e}_i\vct{e}_j^T + \vct{e}_j\vct{e}_i^T).
  	\end{align*}
	where $c_4 > 0$ is an absolute constant. Similarly, when $i = j$, we have
	\begin{align*}
		\E \mtx{Z}_{ij}^k
		&\preceq ep\big[c_4(\zeta + \nu)\big]^k k! (\vct{e}_i\vct{e}_i^T)^k \\
		&\preceq ep\big[c_4(\zeta + \nu)\big]^k k! (2\vct{e}_i\vct{e}_i^T + 2\vct{e}_j\vct{e}_j^T + \vct{e}_i\vct{e}_j^T + \vct{e}_j\vct{e}_i^T).
  	\end{align*}
	Putting them together, we obtain that when $i \le j$,
	$$
	\E \mtx{Z}_{ij}^k \preceq ep\big[c_4(\zeta + \nu)\big]^k k! (2\vct{e}_i\vct{e}_i^T + 2\vct{e}_j\vct{e}_j^T + \vct{e}_i\vct{e}_j^T + \vct{e}_j\vct{e}_i^T).
	$$
	The above inequality implies that the sequence $\{\mZ_{ij}\}$ satisfies the condition of Matrix Bernstein inequality (Lemma \ref{th: matrix bernstein inequality}). Moreover, we have
	$$
	R = c_4(\zeta + \nu),
	$$
	and
	\begin{align*}
		\sigma^2 &= \Big\|\sum_{i \le j} 2ep \big[c_4(\zeta + \nu)\big]^2 (2\vct{e}_i\vct{e}_i^T + 2\vct{e}_j\vct{e}_j^T + \vct{e}_i\vct{e}_j^T + \vct{e}_j\vct{e}_i^T) \Big\| \\
		&= 2ep \big[c_4(\zeta + \nu)\big]^2 \Big\|\sum_{i \le j}(2\vct{e}_i\vct{e}_i^T + 2\vct{e}_j\vct{e}_j^T + \vct{e}_i\vct{e}_j^T + \vct{e}_j\vct{e}_i^T) \Big\| \\
		&\le 2ep \big[c_4(\zeta + \nu)\big]^2 \Big\|\sum_{1 \le i, j \le n}(2\vct{e}_i\vct{e}_i^T + 2\vct{e}_j\vct{e}_j^T + \vct{e}_i\vct{e}_j^T + \vct{e}_j\vct{e}_i^T) \Big\| \\
		&\le  2ep \big[c_4(\zeta + \nu)\big]^2 \cdot 6n \\
		&= c_5np(\zeta+\nu)^2.
  	\end{align*}
	Thus, applying the matrix Bernstein's inequality (Lemma \ref{th: matrix bernstein inequality}) to
	$
	\sum_{i \le j} \mtx{Z}_{ij},
	$
	we obtain that for any $u \ge 0$,
	\begin{multline*}
		\P \Big\{ \Big\| \sum_{i \le j} \mtx{Z}_{ij} \Big\| \ge u \Big\} \le \\
		2n \exp \Big[ -c \cdot \min \Big( \frac{u^2}{np(\zeta+\nu)^2}, \frac{u}{\zeta+\nu} \Big) \Big].
  	\end{multline*}
	Letting $u=\frac{pt}{2\sqrt{2r}}$, we have
	\begin{align*}
		& \P \Big\{ 2\sqrt{2r}p^{-1} \Big\| \sum_{i \le j} \mtx{Z}_{ij} \Big\| \le t \Big\} \\
		& \ge 1-2n \exp \Big[ -c \cdot \min \Big( \frac{pt^2}{rn(\zeta+\nu)^2}, \frac{pt}{\sqrt{r}(\zeta+\nu)} \Big) \Big].
  	\end{align*}
	Recalling that $\| \hat{\mtx{D}}-\mtx{D} \|_F \le 2\sqrt{2r}p^{-1} \Big\| \sum_{i \le j} \mtx{Z}_{ij} \Big\|$, we have
	\begin{align*}
	  	&\P \Big\{ \| \hat{\mtx{D}}-\mtx{D} \|_F \ge t \Big\} \\
		&\le \P \Big\{ 2\sqrt{2r}p^{-1} \Big\| \sum_{ij} \mtx{Z}_{ij} \Big\| \ge t \Big\} \\
		&\le 2n \exp \Big[ -c \cdot \min \Big( \frac{pt^2}{rn(\zeta+\nu)^2}, \frac{pt}{\sqrt{r}(\zeta+\nu)} \Big) \Big] \\
		&= 2n \exp \Big[ -c \cdot \min \Big( \frac{mt^2}{rn^3(\zeta+\nu)^2}, \frac{mt}{n^2\sqrt{r}(\zeta+\nu)} \Big) \Big].
	\end{align*}
	The last equality comes from the definition of $m=pn^2$. This completes the proof.
\end{proof}

\subsection{Error Bound for Gram Matrices after MDS}

Once we have established the error bounds for EDM completion, it is convenient to utilize them to derive a bound for the coordinate recovery error. We adopt the following metric to measure the reconstruction error \cite{Oh2010}:
\begin{equation} \label{eq: reconstruction error metric}
	\dist(\mtx{X}, \hat{\mtx{X}}) = \frac{1}{n} \| \mtx{J} \mtx{X}^T \mtx{X} \mtx{J} - \mtx{J} \hat{\mtx{X}}^T \hat{\mtx{X}} \mtx{J} \|_F,
\end{equation}
where $\mtx{X}$ is the true coordinate matrix having each sensor coordinate as a column, $\hat{\mtx{X}}$ is the estimated coordinate matrix having each estimated sensor coordinate as a column, and $\mtx{J}$ is the geometric centering matrix defined in (\ref{centering_matrix}). Notice that the distances have lost some information such as orientation, since rigid transformation (rotation, reflection and translation) does not change the pairwise distances. We choose \eqref{eq: reconstruction error metric} as the metric of recovery error because it has the following property: (a) It is invariant under rigid transformation; (b) $\dist(\mtx{X}, \hat{\mtx{X}}) = 0$ implies that $\mtx{X}$ and $\hat{\mtx{X}}$ is equivalent up to an unknown rigid transformation. Then we have following result.

\begin{theorem} [Coordinate recovery error] \label{th: reconstruction error}
    Let $\hat{\mtx{X}}$ be the estimated location matrix by SVD-MDS. Then the reconstruction error has the following upper bound:
	\begin{equation} \label{eq: reconstruction error}
	  	\E~\dist(\mtx{X}, \hat{\mtx{X}}) \le C \sqrt{\frac{dn}{m}} (\zeta + \nu),
	\end{equation}
	where $d$ denotes the dimension of the space in which nodes live and $C$ is a constant number same as Theorem \ref{th: expected form of matrix completion accuracy}.
\end{theorem}

\begin{remark}
 	We mention that our Theorems \ref{th: expected form of matrix completion accuracy}--\ref{th: reconstruction error} also hold when we the noise is sub-Gaussian, i.e., the entries on or above the diagonal of $\mEn$ are independent, identical distributed sub-Gaussian random variables with sub-Gaussian norm $\nu$. The proof techniques are essentially the same.
\end{remark}

\begin{proof}
  	The proof of Theorem \ref{th: reconstruction error} is motivated by \cite{Oh2010}. By definition of $\dist(\mtx{X}, \hat{\mtx{X}})$, we have
	\begin{align*}
	    \dist(\mtx{X}, \hat{\mtx{X}})
		&= \big\| \mtx{J} ( \mtx{X}^T\mtx{X} - \hat{\mtx{X}}^T \hat{\mtx{X}} )\mtx{J} \big\|_F \\
		&\le \sqrt{2d} \big\| \mtx{J} ( \mtx{X}^T\mtx{X} - \hat{\mtx{X}}^T \hat{\mtx{X}} )\mtx{J} \big\|,
  	\end{align*}
	where we have used the fact that the matrix $\mtx{J}\mtx{X}^T\mtx{X}\mtx{J}-\mtx{J}\hat{\mtx{X}}^T\hat{\mtx{X}}\mtx{J}$ has rank at most $2d$. Let $\mM = -\frac{1}{2}\mtx{J}\hat{\mtx{D}}\mtx{J}$. The spectral norm can be bounded by the triangle inequality as follows:
	$$
	\big\| \mtx{J} ( \mtx{X}^T\mtx{X} - \hat{\mtx{X}}^T \hat{\mtx{X}} )\mtx{J} \big\| \le \big\| \mtx{J} \mtx{X}^T\mtx{X} \mtx{J} - \mM \big\| + \big\| \mM - \mtx{J}\hat{\mtx{X}}^T \hat{\mtx{X}} \mtx{J} \big\|.
	$$
	Recall that the EDM has the following property \cite{Oh2010}:
	$$
	-\frac{1}{2}\mJ \mD \mJ = \mJ \mX^T \mX \mJ.
	$$
	Thus, the first term can be written as
	$$
	\big\| \mtx{J} \mtx{X}^T\mtx{X} \mtx{J} - \mM \big\| = \Big\| - \frac{1}{2}\mtx{J} \mtx{D} \mtx{J} + \frac{1}{2}\mtx{J}\hat{\mtx{D}}\mtx{J} \Big\| = \frac{1}{2} \big\|\mtx{J} ( \hat{\mtx{D}} - \mtx{D} ) \mtx{J} \big\|.
	$$
	By submultiplicity of spectral norm and the fact that $\|\mtx{J}\|=1$, the right-hand side is bounded by
	$$
	\frac{1}{2} \big\|\mtx{J} ( \hat{\mtx{D}} - \mtx{D} ) \mtx{J} \big\| \le \frac{1}{2} \|\mtx{J}\| \| \hat{\mtx{D}} - \mtx{D} \| \| \mtx{J}\| = \frac{1}{2} \| \hat{\mtx{D}} - \mtx{D} \|.
	$$
	To bound the second term, note that $\hat{\mtx{X}}^T\hat{\mtx{X}}$ is the best rank-$d$ approximation to $\mtx{M}$. Thus, for any rank-$d$ matrix $\mtx{A}$, we have
	$$
	\| \mtx{M} - \hat{\mtx{X}}^T\hat{\mtx{X}} \| \le \| \mtx{M} - \mtx{A} \|.
	$$
	Now the second term can be bounded as follows:
	\begin{align*}
		\big\| \mM - \mtx{J}\hat{\mtx{X}}^T \hat{\mtx{X}} \mtx{J} \big\|
		&= \big\| \mM - \hat{\mtx{X}}^T\hat{\mtx{X}} \big\| \le \big\| \mM + \frac{1}{2} \mtx{J} \mtx{D} \mtx{J} \big\| \\
		&= \frac{1}{2} \big\|\mtx{J} ( \hat{\mtx{D}} - \mtx{D} ) \mtx{J} \big\| \le \frac{1}{2} \| \mtx{D} - \hat{\mtx{D}} \|,
  	\end{align*}
	where the first equality uses the fact that $\mtx{J}\hat{\mtx{X}}^T = \hat{\mtx{X}}^T$ (see \cite[pp. 3]{Oh2010}), the first inequality comes from letting $\mtx{A} = -\frac{1}{2} \mtx{J}\mtx{D}\mtx{J}$, and the last inequality uses again the submultiplicity of spectral norm and $\|\mtx{J}\|=1$. Combining these two terms together, we have
	$$
	\big\| \mtx{J}\mtx{X}^T\mtx{X}\mtx{J} - \mtx{J}\hat{\mtx{X}}^T\hat{\mtx{X}}\mtx{J} \big\| \le \| \mtx{D} - \hat{\mtx{D}} \|.
	$$
	The conclusions follows by a simple application of Theorem \ref{th: expected form of matrix completion accuracy}:
	$$
	\E \dist(\mtx{X}, \hat{\mtx{X}}) \le \sqrt{2d} \E \| \mtx{D} - \hat{\mtx{D}} \| \le C \sqrt{\frac{dn}{m}}(\zeta + \nu).
	$$
\end{proof}

\subsection{Comparison and Discussion}

In this section, we make some comparisons between our results and related works in the literature. First, we compare our results (Theorem \ref{th: expected form of matrix completion accuracy}) with that established in \cite{Drineas2006} for the SVD-Reconstruct approach, and show that our error bound has faster descent rate than that in \cite{Drineas2006}. Next, we compare our Theorem \ref{th: tail bound for matrix completion error} with the results in \cite{Ding2017}, where the authors established a high probability error bound for  a semidefinite programming approach, and conclude that the two results have the same order of error rate.


\subsubsection{Comparison with \cite{Drineas2006}}
We begin by comparing our results (Theorem \ref{th: tail bound for matrix completion error}) with that in \cite{Drineas2006}. For convenience, we restate the results in \cite{Drineas2006} as the following proposition:

\begin{proposition} [Theorem 2, \cite{Drineas2006}] \label{th: original result}
  	Let $\tmD$ be the estimated distance matrix by SVD-Reconstruct without the symmetrization step. Then, with probability at least $1-1/(2n)$,
	\begin{equation} \label{eq: original result}
		\| \mD - \tmD \|_F \le 12 \sigma_S \sqrt{2n} + 8 \sqrt{\sigma_S \sqrt{2n} \|\mD\|_F },
  	\end{equation}
	where $\sigma^2_S$ denotes an upper bound for the variance of the entries of $\mD$, and is bounded by
	\begin{equation} \label{eq: original bound for variance}
		\sigma^2_S \le \frac{2}{p} \max_{i,j} \big( \mD_{ij}^2 + \nu^2 \big).
  	\end{equation}
\end{proposition}
Note that in (\ref{eq: original result}) there is no $r$ appearing because the authors fixed that $r=4$, and in (\ref{eq: original bound for variance}) there is no $\gamma_{ij}$ as it is assumed that $\gamma_{ij}=0$ in this paper.

To compare with our bound, substituting (\ref{eq: original bound for variance}) into (\ref{eq: original result}) and noting that $ \| \mD \|_F \le n \zeta$, we see that \eqref{eq: original result} means that
$$
\frac{1}{n} \| \mD - \tmD \|_F \le C_1 \frac{\zeta+\nu}{\sqrt{np}} + C_2 \frac{\zeta+\nu}{\sqrt[4]{np}} \le C_3 \frac{\zeta+\nu}{\sqrt[4]{np}}.
$$
The last inequality holds because we have assumed that $np \ge \log n > 1$ when $n$ is large enough (e.g., $n \geq 3$). As a result, Proposition \ref{th: original result} implies
\begin{equation} \label{eq: error bound in original}
	\P \Big\{ \frac{1}{n}\| \hat{\mtx{D}}-\mtx{D} \|_F \le C_3 \frac{\zeta+\nu}{\sqrt[4]{np}} \Big\} \ge 1- \frac{1}{2n}.
\end{equation}

In our Theorem \ref{th: tail bound for matrix completion error}, when $t \le (\zeta + \nu) n\sqrt{r}$, the tail bound becomes
$$
\P \Big\{ \| \hat{\mtx{D}}-\mtx{D} \|_F \ge t \Big\} \le n \exp \Big[ -\frac{cmt^2}{n^3r(\zeta+\nu)^2} \Big],
$$
Assume that $m \ge C n \log n$. Then, choosing
$$
t = (\zeta+\nu)n\sqrt{r}\cdot \sqrt{\frac{C n \log n}{m}} \le (\zeta+\nu) n\sqrt{r},
$$
we obtain
\begin{multline*}
\P \Big\{ \frac{1}{n}\| \hat{\mtx{D}}-\mtx{D} \|_F \ge (\zeta+\nu) \sqrt{\frac{Crn\log n}{m}}\Big\} \\
\le n \exp \Big( -\frac{cC \log n}{2}\Big).
\end{multline*}
If $C$ is sufficiently large such that $cC \ge 4$, the last term is bounded by $1/n$. Therefore,
\begin{equation} \label{eq: error with probability 1/n}
\P \Big\{ \frac{1}{n}\| \hat{\mtx{D}}-\mtx{D} \|_F \le (\zeta+\nu) \sqrt{\frac{Cr\log n}{np}}\Big\} \ge 1- \frac{1}{n}.
\end{equation}

Now comparing  \eqref{eq: error bound in original} with \eqref{eq: error with probability 1/n}, we conclude that our bound has faster descent rate (with respect to $n$) than that in \cite{Drineas2006}, provided $np \ge (\log n)^2$.

\subsubsection{Comparison with \cite{Ding2017}}
In the paper \cite{Ding2017}, Ding and Qi proposed a semidefinite programming approach to recover EDMs. They also presented a performance analysis for their method. Their theoretical results (Theorem 1) state that the following event
\begin{equation*} 
\frac{\| \hat{\mtx{D}}-\mtx{D} \|^2_F}{\Delta} \le C ( C_4\zeta^2 + C_5\nu^2\omega^2) \frac{r\Delta \log(2n)}{nm'}
\end{equation*}
holds with probability at least $1-\frac{2}{n}$, where $\Delta = n(n-1)/2$, $\omega^2$ is the maximum of $\vct{\Omega} \odot \vct{D}$, and $m'$ denotes the number of observed entries. Substituting the definitions of $\Delta$ and $\omega$ into the above relation, we obtain
\begin{equation} \label{eq: results of DQ16}
  \P \Big\{ \frac{1}{n}\| \hat{\mtx{D}}-\mtx{D} \|_F \le C_6 \sqrt{\zeta^2+\nu^2\omega^2} \sqrt{\frac{rn\log n}{m'}}\Big\} \ge 1- \frac{1}{n}.
\end{equation}
Considering the observation model in this paper, we may have $m' \approx m = pn^2$. Comparing \eqref{eq: error with probability 1/n} and \eqref{eq: results of DQ16}, if we ignore the effect of $\omega^2$, the two results are of the same order of error rate. It is worth pointing out that the recovery procedure in \cite{Ding2017} is totally different from the SVD-MDS approach, and the former is much more computationally expensive than the latter.

\section{Minimax Lower Bound}

Given a statistical estimation problem, one may want to find the ``best'' estimator. For this purpose, the following question is worth discussing: what is the best performance for a certain estimation problem achievable by any procedure? The {\em minimax risk} answers this question. By definition, it refers to the smallest {\em worst case error} among all estimators. Thus, it describes the intrinsic property of the estimation problem under consideration, and gives us insights into fundamental limitations of performance. In this section, we derive the minimax lower bound of the zero-diagonal symmetric matrix completion problem, and provide a way to evaluate the performance of a certain matrix completion procedure.

In the zero-diagonal symmetric matrix completion scheme, we observe an incomplete and noisy matrix $\mtx{Y}$ and aim to find an estimator which maps our observation $\mtx{Y}$ to an estimate $\hat{\mtx{D}}$ of the true $\mtx{D}$. The performance (or estimate error) is measured via the mean Frobenius norm $\E \| \hat{\mtx{D}} - \mtx{D} \|_F$. Note that given an estimator, the estimate error relies on the true matrix $\mtx{D}$, which is fixed but unknown. Let $\mathcal{D}(r)$ denote the set of all zero-diagonal symmetric matrices with rank at most $r$. Thus, a reasonable performance measure for a certain estimator $\hat{\mtx{D}}$ is the worst case error:
$$
\sup_{\mtx{D} \in \mathcal{D}(r)}~ \E \| \hat{\mtx{D}}-\mtx{D} \|_F.
$$
Then an optimal estimator under this metric gives the minimax risk, which is defined as
\begin{equation} \label{eq: minimax definition}
  	\mathfrak{R} := \inf_{\hat{\mtx{D}}} \sup_{\mtx{D} \in \mathcal{D}(r)}~ \E \| \hat{\mtx{D}}-\mtx{D} \|_F,
\end{equation}
where we take the supremum over all zero-diagonal symmetric matrices with rank at most $r$ and the infimum is taken over all estimators $\hat{\mtx{D}}$.


The minimax problem (\ref{eq: minimax definition}) cannot be solved in general. We instead try to derive a lower bound on the minimax risk $\mathfrak{R}$ using Fano's method \cite{Tsybakov2009}. The lower bound on the minimax risk $\mathfrak{R}$ allows us to evaluate the performance of a certain estimator. In other words, if the worst case error of a given estimator is on the same order of the minimax risk $\mathfrak{R}$, we can say that the estimator is optimal under this metric. Therefore, efforts paid to search estimators with better performance will become meaningless.

\begin{theorem} \label{th: minimax lower bound for symmetric matrix completion}
  The minimax lower bound for the zero-diagonal symmetric low-rank matrix completion problem is
	\begin{equation} \label{eq: minimax lower bound for symmetric matrix completion}
	  	\mathfrak{R} \ge \frac{n\nu}{64} \sqrt{\frac{rn}{m}}.
	\end{equation}
\end{theorem}

\begin{remark}
   Note that the zero-diagonal symmetric low-rank matrix completion problem is not equivalent to the EDM completion problem. Since a zero-diagonal symmetric low-rank matrix may not be an EDM. Nevertheless, our result (Theorem \ref{th: minimax lower bound for symmetric matrix completion}) sheds considerable light on the EDM completion problem, since any EDM must be a zero-diagonal symmetric low-rank matrix.
\end{remark}

As noted in Remark \ref{rm: two cases for noise level}, the SVD-Reconstruct approach only achieve the minimax rate up to a constant factor under high noise level ($\nu \ge \zeta$). When the noise level is low ($\nu < \zeta$), this approach is minimax suboptimal. Therefore, in practical applications, when the noise level is low, it is sensible to develop better approaches to complete the EDM.

To verify this observation, we compare SVD-Reconstruct with another matrix completion algorithm: OptSpace \cite{Keshavan2010}. In our simulation, the points are in $d=3$ dimensional space and the number of points is set to be $50$. Each coordinate of points uniformly distributes in $(-1,1)$. To generate the incomplete EDM, we set the observation probability $p = 0.5$. After that, we add zero-mean Gaussian noise. We compare the performance of SVD-Reconstruct and OptSpace under different noise levels. For each noise level, we run $20$ trials and take the average completion error as a metric for comparison. The completion error is calculated via the Frobenius norm of the error matrix.

\begin{figure}
  	\centering
	\includegraphics[width=.52\textwidth]{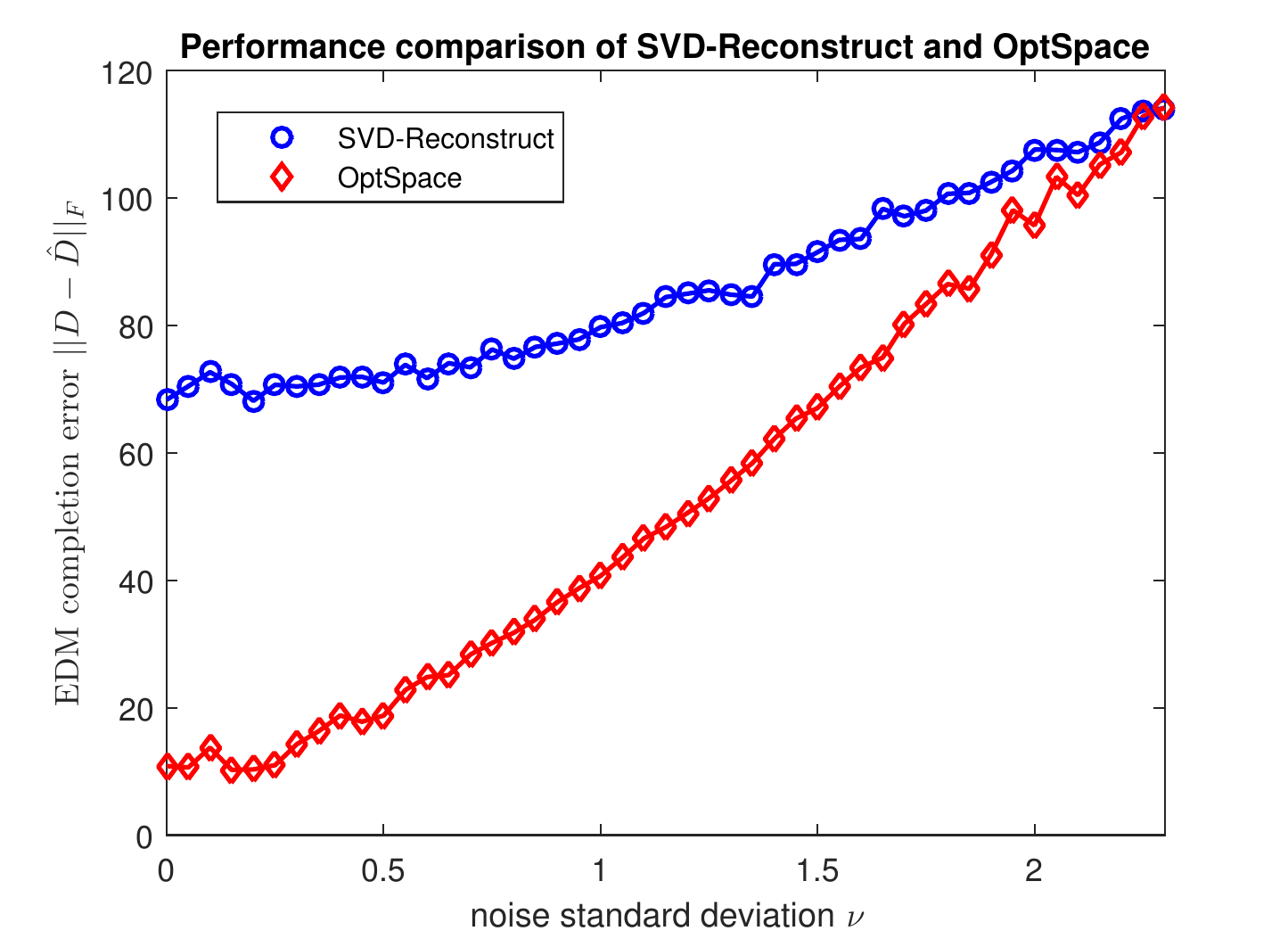}
	\caption{Performance comparison of SVD-Reconstruct and OptSpace under low noise levels.}
	\label{fig: simulation}	
\end{figure}

Figure \ref{fig: simulation} illustrates our results. We can see that when the standard deviation of noise is low (i.e., below $\nu = 2.3$ approximately), the average completion error of SVD-Reconstruct is larger than that of OptSpace. The simulation result agrees with our theoretical analysis that SVD-Reconstruct is suboptimal under low noise levels.

\begin{proof} [Proof of Theorem \ref{th: minimax lower bound for symmetric matrix completion}]
	The proof is divided into two steps. First, we reduce the estimation problem to a hypothesis testing problem. The point lies in that the minimax risk can be lower bounded by the probability of error in a hypothesis testing problem. Once this is done, we will use some well-established tools (e.g., Fano's inequality) in information theory to establish a lower bound on the error probability of the testing problem, and obtain the lower bound for the minimax rate. Let us now give a detailed proof. \\
	\textbf{Step 1: From estimation to testing. } In this step, we consider a finite subset $\mathcal{D}_0 \subset \mathcal{D}(r)$ which is a $\delta$-packing of $\mathcal{D}(r)$. Here, $\delta$-packing means that for any two distinct matrix $\mtx{D}_u$, $\mtx{D}_v \in \mathcal{D}_0$,
	\begin{equation*}
  		\| \mtx{D}_u - \mtx{D}_v \|_F \ge \delta.
	\end{equation*}
	See Fig. \ref{fig: minimax packing} for an example. Given the $\delta$-packing set, we assume that the complete true matrix $\mtx{D}$ in (\ref{eq: symmetric matrix completion model}) is taken {\em uniformly at random} (u.a.r.) from $\mathcal{D}_0$.

	\begin{figure}
  		\centering
		\includegraphics[width=.5\textwidth]{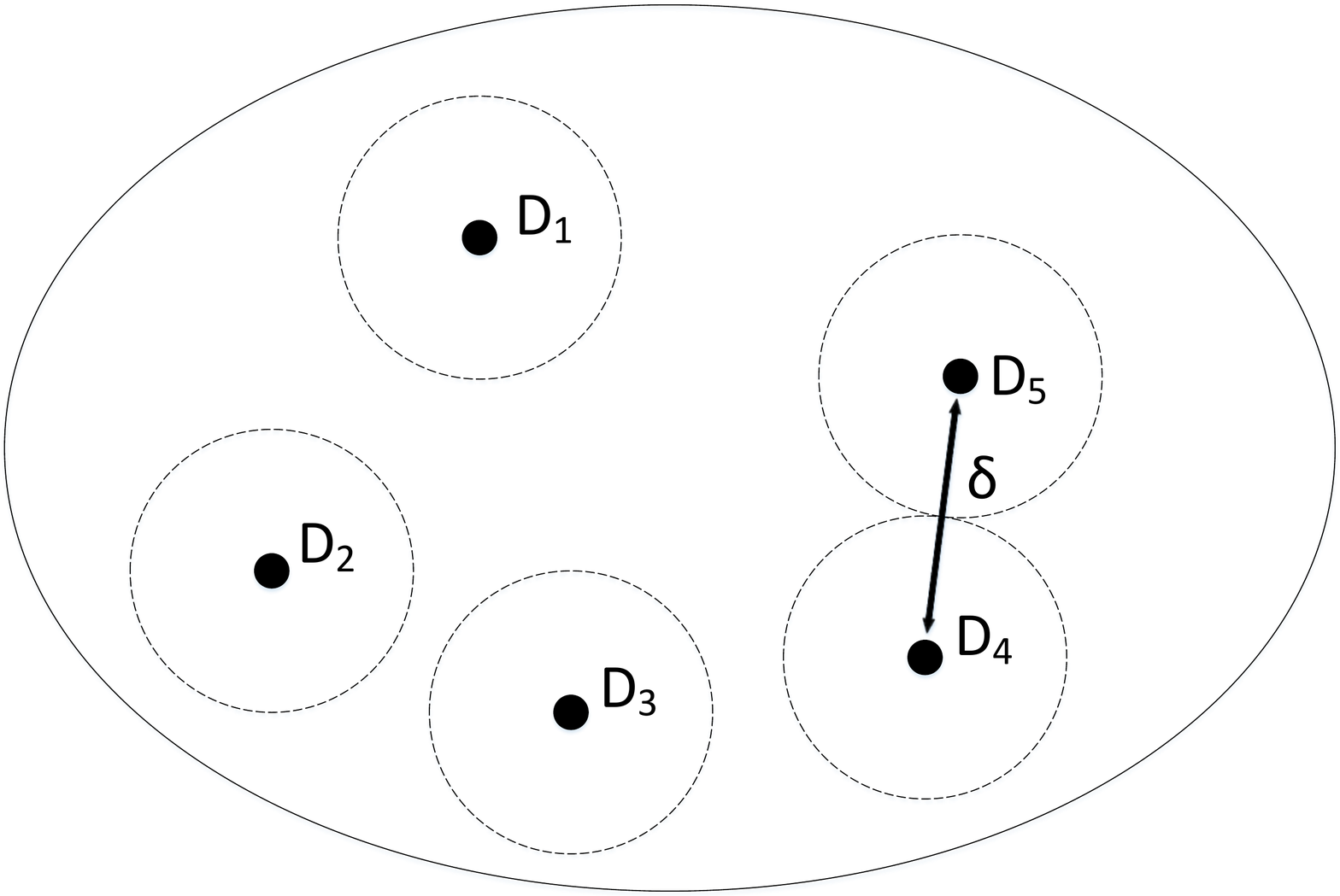}
		\caption{Example of a $\delta$-packing of a set. Here, $\mathcal{D}_0 = \{ D_1, D_2, D_3, D_4, D_5 \}$ } \label{fig: minimax packing}
	\end{figure}

	Now with these setup, we obtain the following hypothesis testing problem: Given the observed matrix $\mtx{Y}$, we need to determine which $\mtx{D}_u \in \mathcal{D}_0$ is the groundtruth $\mtx{D}$. Note that here we do not specify a test function $\Phi: \mtx{Y} \rightarrow \mtx{D}_u \in \mathcal{D}_0$. Lemma \ref{th: from estimation to testing} gives us the relationship between minimax risk $\mathfrak{R}$ and the error probability $\P \{ \Phi (\mtx{Y}) \neq \mtx{D} \}$:
	$$
	\mathfrak{R} \ge \frac{\delta}{2} \inf_{\Phi} \P \{ \Phi (\mtx{Y}) \neq \mtx{D} \},
	$$	
	where the infimum is taken over all testing functions. 
	It remains to bound the error probability $\P \{ \Phi (\mtx{Y}) \neq \mtx{D} \}$. We will use Fano's inequality to do this. \\
	\textbf{Step 2: Bound error probability.}
	Fano's inequality (Lemma \ref{th: Fano inequality} and \ref{th: Fano inequality, uniform}) gives a lower bound for the error probability. But if we want to apply Lemma \ref{th: Fano inequality, uniform} in our zero-diagonal symmetric low-rank matrix completion model (\ref{eq: symmetric matrix completion model}), we must consider the effect of $\mtx{\Omega}$. Different from other estimation problems, here $\mtx{\Omega}$ is random and unfixed. As a result, Lemma \ref{th: Fano inequality, uniform} cannot be applied directly unless we condition on $\mtx{\Omega}$:
	\begin{equation*} \label{eq: Fano inequality for our model}
	  	\P (\hat{\mtx{D}} \neq \mtx{D} | \mtx{\Omega}) \ge 1- \frac{I(\mtx{D};\mtx{Y}|\mtx{\Omega})+\log 2}{\log |\mathcal{D}_0|}.
	\end{equation*}
	By the law of total probability and taking expectation with respect to $\mtx{\Omega}$ on both sides of the above inequality, we get
	\begin{equation*} \label{eq: error probability under Fano inequality}
	  	\P (\hat{\mtx{D}} \neq \mtx{D}) = \E_{\mtx{\Omega}} \P (\hat{\mtx{D}} \neq \mtx{D} | \mtx{\Omega}) \ge 1- \frac{\E_{\mtx{\Omega}} I(\mtx{D};\mtx{Y}|\mtx{\Omega})+\log 2}{\log |\mathcal{D}_0|}.
	\end{equation*}
	Substituting this into Lemma \ref{th: from estimation to testing}, the minimax rate $\mathfrak{R}$ has the following lower bound:
	\begin{equation} \label{eq: Fano's method for our model}
	  	\mathfrak{R} \ge \frac{\delta}{2} \Big( 1- \frac{\E_{\mtx{\Omega}} I(\mtx{D};\mtx{Y}|\mtx{\Omega})+\log 2}{\log |\mathcal{D}_0|} \Big).
	\end{equation}

	Let us now discuss how to choose a suitable $\delta$. Intuitively, if we let $\delta \rightarrow 0$, the right-hand side of (\ref{eq: Fano's method for our model}) becomes trivial, as it tends to be zero as well. Therefore, $\delta$ should be sufficiently large. In practice, a common approach is to choose the largest $\delta > 0$ which makes the mutual information small enough, e.g.,

	\begin{equation} \label{eq: bound Fano}
	  	\frac{\E_{\mtx{\Omega}} I(\mtx{D};\mtx{Y}|\mtx{\Omega})+\log 2}{\log |\mathcal{V}|} \le \frac{1}{2}.
	\end{equation}
	In this case, the minimax risk is lower bounded by $\mathfrak{R} \ge \frac{\delta}{4}$.

	In order to bound the left-hand side of (\ref{eq: bound Fano}), we need a lower bound for the cardinality of the $\delta$-packing set $\mathcal{D}_0$ and an upper bound for the mutual information $I(\mtx{D};\mtx{Y} | \mtx{\Omega})$. For the cardinality of $\mathcal{D}_0$, Lemma \ref{th: size of packing set} establishes the desired result:
	$$
	|D_0| \ge \exp \big( \frac{rn}{128} \big).
	$$

	To establish an upper bound for the mutual information, a typical method involves the {\em Kullback-Leibler divergence} (KL-divergence). Before doing this, let us introduce some necessary notations first. Let $P_{\mtx{Y}}$ denote the distribution of $\mtx{Y}$ conditioned on $\mtx{\Omega}$, $P_u$ denote the distribution of $\mtx{Y}$ conditioned on $\mtx{\Omega}$ and $\mtx{D} = \mtx{D}_u \in \mathcal{D}_0$, and $D_{KL}(f \| g)$ denote the KL-divergence of distribution $f$ and $g$. Then the mutual information can be bounded as follows:
  	\begin{align*} 
  		I(\mtx{D};\mtx{Y}|\mtx{\Omega})
		&= \frac{1}{|\mathcal{D}_0|} \sum_{u:\mtx{D}_u \in \mathcal{D}_0} D_{KL}(P_u \| P_{\mtx{Y}}) \\
		&\le \frac{1}{|\mathcal{D}_0|^2} \sum_{u,v:\mtx{D}_u \in \mathcal{D}_0, \mtx{D}_v \in \mathcal{D}_0} D_{KL}(P_u \| P_v),
	\end{align*}
	where the equality can be derived by a standard calculation, and the inequality is due to the convexity of the $-\log$ function. It remains to calculate the K-L divergence between $P_u$ and $P_v$. Conditioning on $\mO$, either $P_u$ or $P_v$ is a shifted version of the distribution of $\mEn$. Here, we need to be careful, since the noise matrix $\mEn$ is a symmetric matrix, which results in that some of the entries $\mEn_{ij}$'s are dependent. An iterative application of Lemma \ref{th: K-L divergence under dependence} shows that the K-L divergence between $P_u$ and $P_v$ is totally determined by the independent part:
	$$
	D_{KL}(P_u \| P_v) = D_{KL}(P'_u \| P'_v),
	$$
	where $P'_u$ denotes the distribution of $\{ \mY_{ij}: i \le j\}$ conditioning on $\mO$ and $\mD = \mD_u$ and $P'_v$ denotes the distribution of $\{ \mY_{ij}: i \le j\}$ conditioning on $\mO$ and $\mD = \mD_v$.
	
	Since now both $P'_u$ and $P'_v$ are multivariate normal distributions with independent entries of different mean and same variance, the KL-divergence between them can be easily computed:
	\begin{align*}
		D_{KL}(P'_u \| P'_v)
		&= \frac{1}{2\nu^2} \big\| \mtx{\Omega} \odot \mtx{D}_u \odot \mtx{U} - \mtx{\Omega} \odot \mtx{D}_v \odot \mtx{U} \big\|_F^2 \\
		&\le \frac{1}{2\nu^2} \big\| \mtx{\Omega} \odot (\mtx{D}_u - \mtx{D}_v) \big\|_F^2,
	\end{align*}
	where $\mtx{U}$ is a matrix which takes value $1$ on or above the diagonal, and $0$ otherwise. Taking expectation to $\mtx{\Omega}$ both sides, we get
	\begin{align*}
		\E_{\mtx{\Omega}} D_{KL}(P_u \| P_v)
		&\le \frac{1}{2\nu^2} \E_{\mtx{\Omega}} \big\| \mtx{\Omega} \odot (\mtx{D}_u - \mtx{D}_v) \big\|_F^2 \\
		&= \frac{p}{2\nu^2} \big\| \mtx{D}_u - \mtx{D}_v \big\|_F^2.
	\end{align*}
	Thus the minimax risk $\mathfrak{R}$ can be bounded by:
	\begin{align}\label{minimax lower bound step1}
	  	\mathfrak{R}
		&\ge \frac{\delta}{2} \P \{ \Phi (Y) \neq D \} \ge \frac{\delta}{2} \Big( 1- \frac{ \E_{\mtx{\Omega}}I(\mtx{D};\mtx{Y}|\mtx{\Omega})+\log 2}{\log |\mathcal{D}_0|} \Big) \notag \\
		&\ge \frac{\delta}{2} \Bigg[ 1 - \frac{ \E_{\mtx{\Omega}} \Big( \frac{1}{|\mathcal{D}_0|^2} \sum_{u,v} D(P_u \|P_v) \Big) + \log 2}{\log | \mathcal{D}_0 |} \Bigg] \notag \\
		&= \frac{\delta}{2} \Bigg[ 1 - \frac{ \frac{1}{|\mathcal{D}_0|^2} \sum_{u,v} \E_{\mtx{\Omega}} D(P_u \|P_v) + \log 2}{\log | \mathcal{D}_0 |} \Bigg] \notag \\
		&= \frac{\delta}{2} \Bigg[ 1 - \frac{ \frac{p}{2\nu^2|\mathcal{D}_0|^2} \sum_{u,v} \|\mtx{D}_u-\mtx{D}_v\|_F^2 + \log 2}{\log | \mathcal{D}_0 |} \Bigg].
	\end{align}
	Substituting Lemma \ref{th: size of packing set} into (\ref{minimax lower bound step1}), we get
	\begin{align*}
	  	\mathfrak{R}
		&\ge \frac{\delta}{2} \Bigg[ 1- \frac{ \frac{p}{2\nu^2 |\mathcal{D}_0|^2} |\mathcal{D}_0|^2\delta^2 + \log 2}{\frac{rn}{128}} \Bigg] \\
	  	&\ge \frac{\delta}{2} \Bigg[ 1- \frac{ \frac{64p\delta^2}{\nu^2} + 128\log 2}{rn} \Bigg].
	\end{align*}
	Let $\frac{64p\delta^2}{\nu^2} = \frac{rn}{4}$, namely $\delta = \sqrt{ \frac{rn\nu^2}{256p} }$, and assume that $rn \ge 512 \log 2$, then
	\begin{align*}
	  	\mathfrak{R}
	  	&\ge \frac{\delta}{2} \Bigg[ 1- \frac{ \frac{64p\delta^2}{\nu^2} + 128\log 2}{rn} \Bigg] \\
		&\ge \frac{1}{2}\sqrt{ \frac{rn\nu^2}{256p} } \big( 1- \frac{1}{4} - \frac{1}{4} \big) = \frac{\nu}{64} \sqrt{\frac{rn}{p}} = \frac{n\nu}{64} \sqrt{\frac{rn}{m}}.
	\end{align*}
	This completes the proof.
\end{proof}

\section{Conclusion} \label{sec: conclusion}
In this paper, we have presented a detailed analysis for the SVD-MDS approach. We established error bounds for EDM completion by SVD-Reconstruct and for coordinate recovery by MDS using tools from random matrix theory. 
To investigate the optimality of SVD-Reconstruct, we derived the minimax lower bound for the zero-diagonal symmetric low-rank matrix completion problem. The result reveals that when the noise level is high, the SVD-Reconstruct approach can achieve the optimal minimax rate up to a constant factor; when the noise level is low, SVD-Reconstruct is minimax suboptimal, so it is sensible to develop (or employ) more effective methods to complete the EDM and hence localize the positions of nodes.

\appendices

\section{Lemmas Used for Proof of Theorem \ref{th: expected form of matrix completion accuracy}} \label{app: add_a}
  	
\begin{lemma} [Symmetrization, \cite{Ledoux1991}, Lemma~6.3] \label{th: symmetrization}
  Let $F:\R_{+} \rightarrow \R$ be an increasing convex function. Assume that $\mtx{X}_1,\dots,\mtx{X}_N$ are independent, mean zero random vectors in a normed space, and $\epsilon_1,\dots,\epsilon_N$ are independent symmetric Bernoulli random variables. Then
	$$
	\E F \Big( \frac{1}{2}  \Big\| \sum_{i=1}^{N} \epsilon_i \mtx{X}_i \Big\| \Big) \le \E F \Big( \Big\| \sum_{i=1}^{N} \mtx{X}_i \Big\| \Big) \le \E F \Big( 2\Big\| \sum_{i=1}^{N} \epsilon_i \mtx{X}_i \Big\| \Big).
	$$
\end{lemma}

\begin{lemma} [Contraction principle, \cite{Ledoux1991}, Theorem~4.4] \label{th: contraction principle}
  Let $\vct{x}_1,\dots,\vct{x}_N$ be (deterministic) vectors in a normed space, $\epsilon_1,\dots,\epsilon_N$ be independent symmetric Bernoulli random variables, and let $\vct{a}=(a_1,\dots,a_n) \in \R^n$ be a coefficient vector. Then
	$$
	\E \Big\| \sum_{i=1}^{N} a_i \epsilon_i \vct{x}_i \Big\| \le \|a\|_\infty \cdot \E \Big\| \sum_{i=1}^{N} \epsilon_i \vct{x}_i \Big\|.
	$$
\end{lemma}


\begin{lemma} [Seginer's theorem, \cite{Seginer2000}, Theorem~1.1] \label{th: Seginer's theorem}
  There exists a constant $K$ such that, for any $m, n$ any $h \le 2 \log \max \{ m, n \}$ and any $m \times n$ random matrix $\mtx{A}=(a_{ij})$, where $a_{ij}$ are i.i.d. zero mean random variables, the following inequality holds:
  	\begin{eqnarray*}
		&& \max \big\{ \E \max_{1\le i \le m} \|\vct{a}_{i \cdot}\|^h, \E \max_{1 \le j \le n} \|\vct{a}_{\cdot j}\|^h \big\} \le \E \|\mtx{A}\|^h \\
		&& \le (2K)^h \big( \E \max_{1 \le i \le m} \|\vct{a}_{i \cdot}\|^h + \E \max_{1 \le j \le n} \|\vct{a}_{\cdot j}\|^h \big),
	\end{eqnarray*}
	where $\vct{a}_{i \cdot}$ and $\vct{a}_{\cdot j}$ denote the $i$-th row and  the $j$-th column of the matrix $\mtx{A}$, respectively.
\end{lemma}

\begin{lemma} \label{th: Maximum of binomials}
  	(1) Let $X \sim \textnormal{Bin}(n,p)$ and $X_1,\dots,X_n$ be i.i.d. copies of $X$. Let $Z = \max_{i \le n} X_i$. If $p \ge \ln n / n$, then
	$$
	\E Z \le Cnp,
	$$
	where $C$ is an absolute constant. \\
	(2) Let $\mtx{W}$ be an $n \times n$ random matrix with i.i.d Bernoulli entries, i.e., $\P \{ \mtx{W}_{ij} = 1 \} = p$ and $\P \{ \mtx{W}_{ij} = 0 \} = 1-p$, and let $\mtx{G}$ denote an $n \times n$ random matrix whose entries are i.i.d standard Gaussian random variables. Then
	$$
	E \max_{1 \le i \le n} \| (\mtx{W} \odot \mtx{G})_{i \cdot} \|_2 \le c \sqrt{np},
	$$
	where $c$ is an absolute constant.
\end{lemma}
\begin{proof}
  	See Appendix \ref{app: proof of max}.
\end{proof}

\section{Proof of Lemma \ref{th: Maximum of binomials}} \label{app: proof of max}

To prove Lemma \ref{th: Maximum of binomials}, we need the following facts:

\begin{lemma} [Integral identity] \label{th: Integral identity}
  	For any random variable $X$, we have
	$$
	\E X = \int_{0}^{\infty} \P \{ X > t \} dt - \int_{-\infty}^{0} \P \{ X < t \} dt.
	$$
	In particular, for a non-negative random variable $X$, we have
	$$
	\E X = \int_{0}^{\infty} \P \{ X > t \} dt.
	$$
\end{lemma}

\begin{lemma} [Chernoff's inequality, \cite{Michael2005}, Theorem~4.4] \label{th: Chernoff's inequality}
  	Let $X_i$ be independent Bernoulli random variables with parameter $p_i$. Consider their sum $S_N = \sum_{i=1}^{N} X_i$ and denote its mean by $\mu = \E S_N$. Then, for any $t > \mu$, we have
	$$
	\P \{ S_N \ge t \} \le e^{-\mu} \Big( \frac{e \mu}{t} \Big)^t.
	$$
	In particular, for any $t \ge e^2 \mu$ we have
	$$
	\P \{ S_N \ge t \} \le e^{-t}.
	$$
\end{lemma}
\begin{lemma} [Equivalence of sub-Gaussian properties, \cite{Roman2012}, Lemma 5.5] \label{th: equivalence of sub-Gaussian properties}
  	Let $X$ be a random variable. Then the following properties are equivalent with parameters $K_i > 0$ differing from each other by at most an absolute constant factor.

	1. The tails of $X$ satisfy
	$$
	\P \{ |X| > t \} \le \exp (1-t^2/K_1^2) \textnormal{ for all }t \ge 0.
	$$

	2. The moments of $X$ satisfy
	$$
	\|X\|_p = (\E |X|^p)^{1/p} \le K_2 \sqrt{p}.
	$$

	3. The {\em moment generating function} (MGF) of $X$ satisfies
	$$
	\E \exp(X^2/K_3^2) \le 2.
	$$
\end{lemma}
The sub-Gaussian random variable is defined through the above sub-Gaussian properties.
\begin{definition} [Sub-Gaussian random variable, \cite{Roman2012}, Definition 5.7] \label{def: sub-Gaussian}
  	A random variable $X$ that satisfy one of the equivalent properties $1-3$ in Lemma \ref{th: equivalence of sub-exponential properties} is called a sub-Gaussian random variable. The sub-Gaussian norm of $X$, denoted $\|X\|_{\psi_2}$ is defined to be the smallest $K_3$ in property $3$. In other words,
	$$
	\|X\|_{\psi_2} = \inf \{t > 0: \E \exp(X^2/t^2) \le 2 \}.
	$$
\end{definition}
\begin{lemma} [Equivalence of sub-exponential properties, \cite{Roman2012}, pp.~221] \label{th: equivalence of sub-exponential properties}
  	Let $X$ be a random variable. Then the following properties are equivalent with parameters $K_i > 0$ differing from each other by at most an absolute constant factor.

	1. The tails of $X$ satisfy
	$$
	\P \{ |X| > t \} \le \exp (1-t/K_1) \textnormal{ for all }t \ge 0.
	$$

	2. The moments of $X$ satisfy
	$$
	\|X\|_p = (\E |X|^p)^{1/p} \le K_2 p.
	$$

	3. The {\em moment generating function} (MGF) of $X$ satisfies
	$$
	\E \exp(X/K_3) \le 2.
	$$
\end{lemma}
The sub-exponential random variable is defined through the above sub-exponential properties.
\begin{definition} [Sub-exponential random variable, \cite{Roman2012}, Definition~5.13] \label{def: sub-exponential}
  	A random variable $X$ that satisfy one of the equivalent properties $1-3$ in Lemma \ref{th: equivalence of sub-exponential properties} is called a sub-exponential random variable. The sub-exponential norm of $X$, denoted $\|X\|_{\psi_1}$ is defined to be the smallest $K_3$ in property $3$. In other words,
	$$
	\|X\|_{\psi_1} = \inf \{t > 0: \E \exp(X/t) \le 2 \}.
	$$
\end{definition}
The sub-exponential random variables have the following two remarkable properties:
\begin{lemma} [Centering, \cite{Roman2012}, Remark~5.18] \label{th: centering}
  	If $X$ is a sub-exponential random variable then $X - \E X$ is sub-exponential too, and
	$$
	\| X - \E X \|_{\psi_1} \le C \| X \|_{\psi_1},
	$$
	where $C$ is an absolute constant.
\end{lemma}

\begin{lemma} [Product of sub-Gaussian is sub-exponential, \cite{Roman2016}, Lemma~2.7.5] \label{th: product of sub-Gaussian is sub-exponential}
  	Let $X$ and $Y$ be sub-Gaussian random variables. Then $XY$ is sub-exponential. Moreover,
	$$
	\|XY\|_{\psi_1} \le \|X\|_{\psi_2}\|Y\|_{\psi_2}.
	$$
\end{lemma}
  	
\begin{lemma} [Sub-exponential is sub-Gaussian squared, \cite{Roman2012}, Lemma~5.14] \label{th: sub-exponential is sub-Gaussian squared}
  	A random variable $X$ is sub-Gaussian if and only if $X^2$ is sub-exponential. Moreover,
  	$$
	\| X^2 \|_{\psi_1} \le \| X \|^2_{\psi_2}.
	$$
	In particular, if $g \sim N(0, \nu^2)$, we have $g^2$ is sub-exponential, and
	$$
	\| g^2 \|_{\psi_1} \le c \nu^2.
	$$
\end{lemma}


\begin{lemma} [Bernstein-type inequality] \label{th: Modified Bernstein inequality}
  	Let $X$ be a centered random variable satisfying
	$$
	\P \big\{ X \ge t \big\} \le p \exp(-ct),
	$$
	where $0 < p < 1$ and $c$ are constant numbers. Let $X_1,\dots,X_n$ be independent copies of $X$. Then
	\begin{align*}
		\P \Big\{ \sum_{i=1}^{n} X_i \ge t \Big\}
		&\le \exp \Big[ -c \cdot \min \Big( \frac{t^2}{np}, t \Big) \Big] \nonumber \\
		&\le \left\{
	  		\begin{array}{ll}
		  		\exp \big(-\dfrac{ct^2}{np} \big), & t \le np, \\
				\exp \big(-ct \big), & t > np.
			\end{array}
		\right.
	\end{align*}
\end{lemma}

\begin{proof}
  	See Appendix \ref{app:G}.
\end{proof}

Now we are ready to prove Lemma \ref{th: Maximum of binomials}. \\
(1) By integral identity (Lemma \ref{th: Integral identity}), we know that
\begin{align*}
	\E \max_{1 \le i \le n} X_i
	&= \int_{0}^{\infty} \P \{ \max_{1 \le i \le n} X_i > t \} dt \\
	&\le \int_{0}^{s} 1 dt + \int_{s}^{\infty} \P \{ \max_{1 \le i \le n} X_i > t \} dt \\
	&= s + \int_{s}^{\infty} \P \{ \bigcup\limits_{1 \le i \le n} X_i > t \} dt \\
	&\le s + \int_{s}^{\infty} \sum_{i=1}^{n} \P \{ X_i > t \} dt \\
	&= s + \sum_{i=1}^{n} \int_{s}^{\infty} \P \{ X_i > t \} dt.
\end{align*}
If we let $s = e^2np$, by Chernoff's inequality (Lemma \ref{th: Chernoff's inequality}), we have
\begin{align*}
	\E \max_{1 \le i \le n} X_i
	&\le s + \sum_{i=1}^{n} \int_{s}^{\infty} \P \{ X_i > t \} dt \\
	&\le e^2np + \sum_{i=1}^{n} \int_{e^2np}^{\infty} e^{-t} dt \\
	&= e^2np + ne^{-e^2np}.
\end{align*}
Recall that we have assumed that $np \ge \log n$, thus
\begin{align*}
	\E \max_{1 \le i \le n} X_i
	&\le e^2np + ne^{-e^2np} \\
	&\le e^2np + ne^{-e^2 \log n} \\
	&= e^2np + n^{1-e^2} \\
	&\le (e^2+1)np,
\end{align*}
where the last inequality comes from the fact that $n^{-e^2} \le n^{-1} \le \frac{\ln n}{n} \le p$. This completes the proof of the first part of Lemma \ref{th: Maximum of binomials}.

(2) The proof of the second part is similar to that of the first part. 
First,
\begin{align*}
&\E \max_{1 \le i \le n} \| (\mtx{W} \odot \mtx{G})_{i \cdot} \|_2^2 \\
&\le \E \max_{1 \le i \le n} \sum_{j=1}^{n} \mtx{W}_{ij}^2 (\mtx{G}_{ij}^2-1)
+ \E \max_{1 \le i \le n} \sum_{j=1}^{n} \mtx{W}_{ij}^2.
\end{align*}
The second term can be bounded by part (1):
$$
\E \max_{1 \le i \le n} \sum_{j=1}^{n} \mtx{W}_{ij}^2 \le c_1 np.
$$
For the first term, let $\mT_{ij} = \mtx{W}_{ij}^2(\mtx{G}_{ij}^2-1)$, then for any $t > 0$, we have
\begin{align} \label{eq: tail of Tij}
  	\P \big\{ \mT_{ij} \ge t \big\} \nonumber
	&= \P \big\{ \mtx{W}_{ij}^2(\mtx{G}_{ij}^2-1) \ge t | \mtx{W}_{ij} = 1 \big\} \P \{ \mtx{W}_{ij} = 1 \} \nonumber\\
	&\hspace{11.1pt}+ \P \big\{ \mtx{W}_{ij}^2(\mtx{G}_{ij}^2-1) \ge t | \mtx{W}_{ij} = 0 \big\} \P \{ \mtx{W}_{ij} = 0 \} \nonumber\\
	&= p\P \big\{ \mtx{G}_{ij}^2-1 \ge t \big\} \nonumber\\
	&\le p \exp(-ct),
\end{align}
where the last inequality holds because $\mtx{G}_{ij}^2-1 = \mtx{G}_{ij}^2 - \E \mtx{G}_{ij}^2$ is a sub-exponential random variable by Lemma \ref{th: centering} and \ref{th: sub-exponential is sub-Gaussian squared}, and $\|\mtx{G}_{ij}^2-1\|_{\psi_1} \leq c'\|\mtx{G}_{ij}^2\|_{\psi_1} \leq c'\|\mtx{G}_{ij}\|_{\psi_2}^2 \le c''$, where $c'$ and $c''$ are absolute constants. Thus, Lemma \ref{th: Modified Bernstein inequality} can be used to bound the tail probability of $\sum_{j=1}^{n}\mT_{ij}$:
\begin{align} \label{eq: modified Bernstein inequality}
	\P \Big\{ \sum_{j=1}^{n} \mT_{ij} \ge t \Big\}
	&\le \exp \Big[ -c \cdot \min \Big( \frac{t^2}{np}, t \Big) \Big] \nonumber \\
	&\le \left\{
	  	\begin{array}{ll}
		  	\exp \big(-\dfrac{ct^2}{np} \big), & t \le np, \\
			\exp \big(-ct \big), & t > np.
		\end{array}
	\right.
\end{align}
Similarly, we will use integral identity (Lemma \ref{th: Integral identity}) to bound the first term:
\begin{align*}
  	&\E \max_{1 \le i \le n} \sum_{j=1}^{n} \mtx{W}_{ij}^2 (\mtx{G}_{ij}^2-1) \\
	&= \int_{0}^{\infty} \P \big\{ \max_{1 \le i \le n} \sum_{j=1}^{n} \mtx{W}_{ij}^2 (\mtx{G}_{ij}^2-1) \ge t \big\} dt \\
	&\le \int_{0}^{np} 1 dt + n\int_{np}^{\infty} \P \big\{ \sum_{j=1}^{n} \mtx{W}_{ij}^2 (\mtx{G}_{ij}^2-1) \ge t \big\} dt \\
	&\le np + n\int_{np}^{\infty} \exp \big(-ct \big) dt \\
	&\le np + n\exp(-cnp) \\
	&\le c_2np,
\end{align*}
where the last inequality comes from the fact that $np \ge \log n$, and $c_2$ is an absolute constant. Combining the two terms, we get
$$
\E \max_{1 \le i \le n} \| (\mtx{W} \odot \mtx{G})_{i \cdot} \|_2^2 \le c_3 np.
$$
Then, Jensen's inequality completes the proof:
$$
\E \max_{1 \le i \le n} \| (\mtx{W} \odot \mtx{G})_{i \cdot} \|_2 \le \big[ \E \max_{1 \le i \le n} \| (\mtx{W} \odot \mtx{G})_{i \cdot} \|_2^2\big]^{1/2} \le c \sqrt{np}.
$$

\section{Lemmas Used for Proof of Theorem \ref{th: tail bound for matrix completion error}} \label{app: add_b}

\begin{lemma} [Matrix Bernstein inequality: sub-exponential case, \cite{Tropp2012}, Lemma~6.2] \label{th: matrix bernstein inequality}
  	Consider a finite sequence ${\mtx{X}_k}$ of independent, random, self-adjoint matrices with dimension $n$. Assume that each random matrix has zero mean, and satisfies
	$$
	\E \mtx{X}_k = 0 \ \ \textnormal{and} \ \ \E(\mtx{X}_k^p) \preceq \frac{p!}{2}\cdot R^{p-2}\cdot \mtx{A}_k^2 \ \textnormal{for }p = 2,3,4\dots
	$$
	Compute the variance parameter
	$$
  \sigma^2 := \Big\| \sum_k \mtx{A}_k^2 \Big\|.
	$$
	Then the following holds for all $t \ge 0$,
	$$
	\P \Big\{ \Big\| \sum_{k} \mtx{Z}_k \Big\| \ge t \Big\} \le n\cdot \exp \Big[ -c \cdot \min \Big( \frac{t^2}{\sigma^2}, \frac{t}{R} \Big) \Big],
	$$
	where $c$ is an absolute constant.
\end{lemma}

\section{Lemmas Used for Proof of Theorem \ref{th: minimax lower bound for symmetric matrix completion}} \label{app: add_c}

\begin{lemma} [\cite{Tsybakov2009}, pp.~79-80] \label{th: from estimation to testing}
	For any test function $\Phi: \mtx{Y} \rightarrow \mtx{D}_u \in \mathcal{D}_0$, the minimax risk $\mathfrak{R}$ in (\ref{eq: minimax definition}) has the following lower bound:
	\begin{equation*} \label{eq: from estimation to testing}
  		\mathfrak{R} \ge \frac{\delta}{2} \P \{ \Phi (\mtx{Y}) \neq \mtx{D} \}.
	\end{equation*}
\end{lemma}

\begin{lemma} [Fano inequality, \cite{Cover2008}, Theorem 2.10.1] \label{th: Fano inequality}
  	Suppose $V$ is a random variable taking values in a finite set $\mathcal{V}$. For any Markov chain $V \rightarrow X \rightarrow \hat{V}$, we have
	\begin{equation} \label{eq: Fano inequality}
  		h_2 \big( \P ( \hat{V} \neq V ) \big) + \P ( \hat{V} \neq V ) \log (|\mathcal{V}|-1) \ge H(V|\hat{V}),
	\end{equation}
	where the function $h_2(p) = -p\log p - (1-p)\log (1-p)$ denotes the entropy of the Bernoulli random variable with parameter $p$, and $H(V|\hat{V})$ denotes the entropy of $V$ conditioned on $\hat{V}$.
\end{lemma}

Moreover, if we assume that $V$ takes value u.a.r. on $\mathcal{V}$, Lemma \ref{th: Fano inequality} becomes

\begin{lemma} \label{th: Fano inequality, uniform}
	Assume that $V$ is uniform on $\mathcal{V}$. For any Markov chain $V \rightarrow X \rightarrow \hat{V}$,
	\begin{equation*} \label{eq: Fano inequality, uniform}
  		\P ( \hat{V} \neq V ) \ge 1 - \frac{I(V;X)+\log 2}{\log |\mathcal{V}|},
	\end{equation*}
	where $I(V;X)$ denotes the mutual information of random variable $V$ and $X$.
\end{lemma}

\begin{lemma} [Bounding the mutual information, \cite{Cover2008}] \label{th: bounding the mutual information}
  	\begin{align*} 
  		I(\mtx{D};\mtx{Y}|\mtx{\Omega})
		&= \frac{1}{|\mathcal{D}_0|} \sum_{u:\mtx{D}_u \in \mathcal{D}_0} D_{KL}(P_u \| P_{\mtx{Y}}) \\
		&\le \frac{1}{|\mathcal{D}_0|^2} \sum_{u,v:\mtx{D}_u \in \mathcal{D}_0, \mtx{D}_v \in \mathcal{D}_0} D_{KL}(P_u \| P_v).
	\end{align*}
\end{lemma}

The proof of Theorem \ref{th: minimax lower bound for symmetric matrix completion} is based on the following lemma:
\begin{lemma} \label{th: size of packing set}
  	Let $n \ge 10$ be a positive integer, and let $\delta > 0$. Then for each $r = 2,\dots,n$, there exists a set of $n$-dimensional matrices $\{ \mtx{D}^1, \mtx{D}^2, \dots, \mtx{D}^M \}$ with cardinality $M= \lfloor \exp \big( \frac{rn}{128} \big) \rfloor$ such that each matrix is symmetric, with zero diagonal, has rank at most $r$, and moreover
	$$
	\| \mtx{D}^l \|_F = \delta \quad \text{for all }l=1,2,\dots,M,
	$$
	$$
	\| \mtx{D}^k-\mtx{D}^l \|_F \ge \delta \quad \text{for all }k \neq l.
	$$
\end{lemma}

\begin{proof}
  	See Appendix \ref{app: B}.
\end{proof}

\begin{lemma} \label{th: K-L divergence under dependence}
  	Let $p(x_1,x_2)$ denote the \textit{probability density function} (p.d.f) of $X = (X_1, X_2)$, $q(x_1,x_2)$ denote the p.d.f of $X' = (X'_1, X'_2)$, $p(x_1)$ denote the p.d.f of $X_1$, and $q(x_1)$ denote the p.d.f of $X'_1$, where $X_2 = X_1$, $X'_2 = X'_1$. Then, the K-L divergence between $p(x_1,x_2)$ and $q(x_1,x_2)$ is equal to that of $p(x_1)$ and $q(x_1)$:
	$$
	D_{KL}\big(p(x_1,x_2) \| q(x_1,x_2) \big) = D_{KL}\big(p(x_1) \| q(x_1) \big).
	$$
\end{lemma}

\begin{proof}
  	See Appendix \ref{app: D}.
\end{proof}

\section{Proof of Lemma \ref{th: size of packing set}} \label{app: B}

The idea for proof of Lemma \ref{th: size of packing set} is inspired by Negahban and Wainwright \cite{Negahban2012}. It relies on the Hoeffding's inequality, which gives a tail bound for sum of independent Rademaker random variables.

\begin{lemma} [Hoeffding's inequality, \cite{Hoeffding1963}, Theorem~2] \label{th: hoeffding's inequality}
  	Let $X_1$,\dots,$X_N$ be independent symmetric Bernoulli random variables, and $\vct{a} = (a_1,\dots,a_N) \in \R^N$. Then, for any $t > 0$, we have
	$$
	\P \Big\{ \sum_{i=1}^{N} a_iX_i \ge t \Big\} \le \exp \Big( - \frac{t^2}{2\|\vct{a}\|_2^2} \Big).
	$$
\end{lemma}

We proceed via the probabilistic method, in particular by showing that a random procedure succeeds in generating such a set with probability at least $0.22$. Let $M = \lfloor \exp \big( \frac{rn}{128} \big) \rfloor$, and for each $l = 1,2,\dots,M$, we draw a random matrix $ \mtx{\tilde{D}}^l \in \R^{n \times n}$ according to the following procedure:

(a) For rows $i=1,\dots,\lfloor \frac{r}{2} \rfloor$ and columns $j=i+1,\dots,n$, choose each $\mtx{\tilde{D}}^l_{ij} \in \{ -1, +1 \}$ uniformly at random, independently across $(i,j)$.

(b) For columns $j=1,\dots, \lfloor \frac{r}{2} \rfloor$ and rows $i=j + 1,\dots,n$, set $ \mtx{\tilde{D}}^l_{ij} = \mtx{\tilde{D}}^l_{ji}$.

(c) For rows $i=\lfloor \frac{r}{2} \rfloor + 1, \dots, n$ and columns $j=\lfloor \frac{r}{2} \rfloor+1,\dots,n$, and for $1 \le i =j \le n$, set $ \mtx{\tilde{D}}^l_{ij} = 0$.

By construction, each matrix $ \mtx{\tilde{D}}^l $ is symmetric, with zero diagonal, and has rank at most $r$. Since there exists a ceil operator in $\frac{r}{2}$, we go ahead by considering $r$ is even and odd separately. \\
\textbf{Case 1: }$r$ is even and $r \ge 2$. \\
In this case $\lfloor \frac{r}{2} \rfloor = \frac{r}{2}$, and the Frobenius norm $ \| \tilde{\mtx{D}}^l \|_F = \sqrt{r(n-r/4-1/2)}$. We define $\mtx{D}^l = \frac{\delta}{\sqrt{r(n-r/4-1/2)}} \mtx{\tilde{D}}^l$ for all $l=1,\dots,M$. The rescaled matrices $\mtx{D}^l$ has Frobenius norm $\|\mtx{D}^l\|_F = \delta$. We now prove that
$$
\|\mtx{D}^l - \mtx{D}^k\|_F \ge \delta \quad \text{for all }l \neq k
$$
holds with probability at least $0.46$. Now to prove Lemma \ref{th: size of packing set}, it suffices to show that $\| \mtx{\tilde{D}}^l - \mtx{\tilde{D}}^k \|_F \ge \sqrt{r(d-r/4-1/2)}$ with probability at least $0.46$ for any pair $l \neq k$. We have
\begin{align*}
	&\frac{1}{r(n-r/4-1/2)} \| \mtx{\tilde{D}}^l - \mtx{\tilde{D}}^k \|_F^2 \\
	&= \frac{2}{r(n-r/4-1/2)} \cdot \Big[ \sum_{i=1}^{r/2} \sum_{j=i+1}^{n} ( \mtx{\tilde{D}}^l_{ij} - \mtx{\tilde{D}}^k_{ij} )^2 \Big].
\end{align*}
This is a sum of i.i.d. variables, each taking value $\{0, 4\}$ with equal probability, so the Hoeffding's inequality (Lemma \ref{th: hoeffding's inequality}) implies that for any $t \ge 0$,
\begin{align*}
  	&\P \Bigg\{ \frac{2}{r(n-r/4-1/2)} \cdot \\
	&\hspace*{57pt} \bigg\{ \sum_{i=1}^{r/2} \sum_{j=i+1}^{n} \Big[ \frac{1}{2}( \mtx{\tilde{D}}^l_{ij} - \mtx{\tilde{D}}^k_{ij} )^2-1 \Big] \bigg\} \le -t \Bigg\} \\
	&\le \exp \Bigg\{ - \frac{t^2}{2 \cdot \sum\limits_{i=1}^{r/2}\sum\limits_{j=i+1}^{n} \big( \frac{2}{r(n-r/4-1/2)} \big)^2 } \Bigg\} \\
	&\le \exp \Big[ - \frac{r(n-r/4-1/2)t^2}{4} \Big].
\end{align*}
Therefore,
\begin{align*}
	&\P \{ \frac{1}{r(n-r/4-1/2)}\| \mtx{\tilde{D}}^l - \mtx{\tilde{D}}^k \|_F^2 \le 2-t \} \\
	&\le \exp \Big[ - \frac{r(n-r/4-1/2)t^2}{16} \Big] \\
	&\le \exp \Big( - \frac{rnt^2}{32} \Big),
\end{align*}
where in the last inequality we have used the fact that $n \ge r \ge 2$. Since there are less than $M^2$ pairs of matrices in total, by taking union bound we get
$$
\P \{ \min_{l \neq k} \frac{1}{r(n-r/4)}\| \mtx{\tilde{D}}^l - \mtx{\tilde{D}}^k \|_F^2 \le 2-t \} \le M^2 \exp \Big( - \frac{rnt^2}{32} \Big).
$$
Letting $t = 1$ and substituting $M = \lfloor \exp \big( \frac{rn}{128} \big) \rfloor$, we get
\begin{align*}
	&\P \Big\{ \min_{l \neq k} \frac{1}{r(n-r/4)}\| \mtx{\tilde{D}}^l - \mtx{\tilde{D}}^k \|_F^2 \le 1 \Big\} \\
	&\le \exp \Big( \frac{rn}{64} \Big) \exp \Big( - \frac{rn}{32} \Big) \\
	&= \exp \Big( -\frac{rn}{64} \Big).
\end{align*}
Namely,
\begin{align*}
	&\P \Big\{ \min_{l \neq k} \| \mtx{\tilde{D}}^l - \mtx{\tilde{D}}^k \|_F^2 \ge r(n-r/4) \Big\} \\
	&\ge 1-\exp \Big[ -\frac{rn}{64} \Big] \ge 0.26,
\end{align*}
where in the last inequality we have used the fact that $n \ge 10$ and $r \ge 2$. \\
\textbf{Case 2: }$r$ is odd and $r \ge 3$. \\
In the case of $r$ is odd, $\lfloor \frac{r}{2} \rfloor = \frac{r-1}{2}$, and the Frobenius norm $ \| \tilde{\mtx{D}}^l \|_F = \sqrt{(r-1)(n-(r-1)/4-1/2)}$. We define $\mtx{D}^l = \frac{\delta}{\sqrt{(r-1)(n-(r-1)/4-1/2)}} \mtx{\tilde{D}}^l$ for all $l=1,\dots,M$. Similarly as case 1, we will prove that the procedure generates a sequence of matrices satisfying
$$
\| \mtx{\tilde{D}}^l - \mtx{\tilde{D}}^k \|_F \ge \sqrt{(r-1)(n-(r-1)/4-1/2)}
$$
with probability at least $0.2$. We proceed the same as case 1, and obtain that
\begin{align*}
	&\P \{ \min_{l \neq k} \frac{1}{(r-1)(n-(r-1)/4)}\| \mtx{\tilde{D}}^l - \mtx{\tilde{D}}^k \|_F^2 \ge 1 \} \\
	&\ge 1-\exp \Big[ -\frac{5rn}{576} \Big] \ge 0.22,
\end{align*}
where the last inequality holds because $n \ge 10$ and $r \ge 3$.

Combining these two cases and recalling the definition of $\mtx{D}^l$ completes the proof.

\section{Proof of Lemma \ref{th: K-L divergence under dependence}} \label{app: D}
To prove Lemma \ref{th: K-L divergence under dependence}, we need the following two properties of K-L divergence:
\begin{lemma} [information inequality, \cite{Cover2008}, Theorem~2.6.3] \label{th: information inequality}
	Let $p(x), q(x)$ be two probability density functions. Then,
	$$
	D_{KL}\big( p(x) \| q(x) \big) \ge 0
	$$
	with equality if and only if $p(x)=q(x)$ for all $x$.
\end{lemma}
\begin{lemma} [Chain rule for K-L divergence, \cite{Cover2008}, Theorem~2.5.3] \label{th: Chain rule for K-L divergence}
  	Let $p(x_1,x_2)$ and $q(x_1,x_2)$ be the joint p.d.f's of $(X_1,X_2)$ and $(X'_1,X'_2)$, respectively. Denote $p(x_1)$ and $q(x_1)$ the marginal p.d.f's of $X_1$ and $X'_1$, and $p(x_2|x_1)$ and $q(x_2|x_1)$ the conditional p.d.f's of $X_2$ conditioning on $X_1$ and $X'_2$ conditioning on $X'_1$, respectively. Then
  	\begin{multline*}
		D_{KL} \big( p(x_1,x_2) \| q(x_1,x_2) \big) \\
		= D_{KL} \big( p(x_1) \| q(x_1) \big) + \E_{x_1} D_{KL} \big( p(x_2|x_1) \| q(x_2|x_1) \big).
  	\end{multline*}
\end{lemma}
\begin{proof} [Proof of Lemma \ref{th: K-L divergence under dependence}]
  	According to the chain rule for K-L divergence (Lemma \ref{th: Chain rule for K-L divergence}), the K-L divergence between $p(x_1,x_2)$ and $q(x_1,x_2)$ can be written as
	\begin{multline} \label{eq: K-L divergence under dependence}
		D_{KL}\big(p(x_1,x_2) \| q(x_1,x_2)\big) \\
		= D_{KL}\big(p(x_1) \| q(x_1)\big) + \E_{x_1} D_{KL}\big( p(x_2|x_1) \| q(x_2|x_1)\big),
  	\end{multline}
	By Lemma \ref{th: K-L divergence under dependence}, we always have $X_2 = X_1$ and $X'_2 = X'_1$. As a result, the conditional p.d.f's $p(x_2|x_1)$ and $q(x_2|x_1)$ are equivalent, i.e., both of them are delta functions at $x_1$. Hence, the information inequality (Lemma \ref{th: information inequality}) implies that $D_{KL}\big( p(x_2|x_1) \| q(x_2|x_1) \big)=0$. Substituting this into \eqref{eq: K-L divergence under dependence} completes the proof.
\end{proof}

 \section{Proof of Lemma \ref{th: Modified Bernstein inequality}}\label{app:G}

 Lemma \ref{th: Modified Bernstein inequality} is a Bernstein-type inequality, and the proof technique is standard. First, we bound the $n$-moments of $X$:
 \begin{align*}
 	\E |X|^n
 	&= \int_{0}^{\infty} \P \big\{ |X| \ge t \big\} nt^{n-1} dt \\
 	&\le \int_{0}^{\infty} p \exp(-ct) nt^{n-1}dt \\
 	&\le c_1 p n^n.
 \end{align*}
 Then the moment generating function of $X$ can be bounded by
 \begin{align*}
 	\E \exp(\lambda X)
 	&= 1 + |\lambda| \E X + \sum_{k=2}^{\infty}\frac{|\lambda|^k \E X^k}{k!} \\
 	&\le 1 + \sum_{k=2}^{\infty} \frac{|\lambda|^k c_1 p k^k}{k!} \\
 	&\le 1 + c_1 p \sum_{k=2}^{\infty} (e|\lambda|)^k \\
 	&\le 1 + c_2 p \lambda^2 \quad(\text{when }|\lambda| \le 1/2e) \\
 	&\le \exp(c_2 p \lambda^2).
 \end{align*}
 Now the tail probability of $X_i$ can be bounded by
 \begin{align*}
 	\P \big\{ \sum_{i=1}^{n} X_i \ge t \big\}
 	&= \P \big\{ \exp \big( \lambda\sum_{i=1}^{n} X_i \big) \ge \exp(\lambda t) \big\} \\
 	&\le e^{-\lambda t} \prod_i \E \exp(\lambda X_i).
 \end{align*}
 If $|\lambda| \le 1/2e$, then we have
 $$
 \P \big\{ \sum_{i=1}^{n} X_i \ge t \big\} \le \exp(-\lambda t + c_2 np \lambda^2).
 $$
 It remains to optimize over $\lambda > 0$. Choosing $\lambda = \min (t/2c_2 np, 1/2e)$ yields the desired result.


\bibliographystyle{IEEEtran}
\bibliography{IEEEabrv,references}

\begin{thebibliography}{10}
\providecommand{\url}[1]{#1}
\csname url@samestyle\endcsname
\providecommand{\newblock}{\relax}
\providecommand{\bibinfo}[2]{#2}
\providecommand{\BIBentrySTDinterwordspacing}{\spaceskip=0pt\relax}
\providecommand{\BIBentryALTinterwordstretchfactor}{4}
\providecommand{\BIBentryALTinterwordspacing}{\spaceskip=\fontdimen2\font plus
\BIBentryALTinterwordstretchfactor\fontdimen3\font minus
  \fontdimen4\font\relax}
\providecommand{\BIBforeignlanguage}[2]{{%
\expandafter\ifx\csname l@#1\endcsname\relax
\typeout{** WARNING: IEEEtran.bst: No hyphenation pattern has been}%
\typeout{** loaded for the language `#1'. Using the pattern for}%
\typeout{** the default language instead.}%
\else
\language=\csname l@#1\endcsname
\fi
#2}}
\providecommand{\BIBdecl}{\relax}
\BIBdecl

\bibitem{Drineas2006}
P.~Drineas, A.~Javed, M.~Magdon-Ismail, G.~Pandurangant, R.~Virrankoski, and
  A.~Savvides, ``Distance matrix reconstruction from incomplete distance
  information for sensor network localization,'' in \emph{2006 3rd Annual IEEE
  Communications Society on Sensor and Ad Hoc Communications and Networks},
  vol.~2, Sep. 2006, pp. 536--544.

\bibitem{Patwari2005}
N.~Patwari, J.~N. Ash, S.~Kyperountas, A.~O. Hero, R.~L. Moses, and N.~S.
  Correal, ``Locating the nodes: Cooperative localization in wireless sensor
  networks,'' \emph{{IEEE} Signal Process. Mag.}, vol.~22, no.~4, pp. 54--69,
  Jul. 2005.

\bibitem{Havel1985}
T.~F. Havel and K.~W{\"u}thrich, ``An evaluation of the combined use of nuclear
  magnetic resonance and distance geometry for the determination of protein
  conformations in solution,'' \emph{J. Mol. Biol.}, vol. 182, no.~2, pp.
  281--294, Aug. 1985.

\bibitem{Dokmanic2013}
I.~Dokmani{\'c}, R.~Parhizkar, A.~Walther, Y.~M. Lu, and M.~Vetterli,
  ``Acoustic echoes reveal room shape,'' \emph{Proc. Natl. Acad. Sci.}, vol.
  110, no.~30, pp. 12\,186--12\,191, Jun. 2013.

\bibitem{Weinberger2004}
K.~Q. Weinberger and L.~K. Saul, ``Unsupervised learning of image manifolds by
  semidefinite programming,'' \emph{Proc. {IEEE} Conf. on Computer Vision and
  Pattern Recognition}, vol.~2, pp. II--988--II--995, 2004.

\bibitem{Torgeson1965}
W.~S. Torgeson, ``Multidimensional scaling of similarity,''
  \emph{Psychometrika}, vol.~30, no.~4, pp. 379--393, Dec. 1965.

\bibitem{Gower1985}
J.~C. Gower, ``Properties of euclidean and non-euclidean distance matrices,''
  \emph{Linear Algebra Appl.}, vol.~67, pp. 81--97, Jun. 1985.

\bibitem{Gower1982}
------, ``Euclidean distance geometry,'' \emph{Math. Sci.}, vol.~7, no.~1, pp.
  1--14, Jan. 1982.

\bibitem{Parhizkar2013a}
R.~Parhizkar, A.~Karbasi, S.~Oh, and M.~Vetterli, ``Calibration using matrix
  completion with application to ultrasound tomography,'' \emph{{IEEE} Trans.
  Signal Process.}, vol.~61, no.~20, pp. 4923--4933, Oct. 2013.

\bibitem{Keshavan2010}
R.~H. Keshavan, A.~Montanari, and S.~Oh, ``Matrix completion from noisy
  entries,'' \emph{{IEEE} Trans. Inf. Theory}, vol.~56, no.~6, pp. 2980--2998,
  Jun. 2010.

\bibitem{Alfakih1999}
A.~Y. Alfakih, A.~Khandani, and H.~Wolkowicz, ``Solving euclidean distance
  matrix completion problems via semidefinite programming,'' \emph{Comput.
  Optim. Appl.}, vol.~12, no.~1, pp. 13--30, Jan. 1999.

\bibitem{Biswas2006}
P.~Biswas, T.~Liang, K.~Toh, Y.~Ye, and T.~Wang, ``Semidefinite programming
  approaches for sensor network localization with noisy distance
  measurements,'' \emph{{IEEE} Trans. Autom. Sci. Eng.}, vol.~3, no.~4, pp.
  360--371, Oct. 2006.

\bibitem{Javanmard2013}
A.~Javanmard and A.~Montanari, ``Localization from incomplete noisy distance
  measurements,'' \emph{Found. Comput. Math.}, vol.~13, no.~3, pp. 297--345,
  2013.

\bibitem{Ding2017}
C.~Ding and H.~Qi, ``Convex optimization learning of faithful euclidean
  distance representations in nonlinear dimensionality reduction,''
  \emph{Mathematical Programming}, vol. 164, no.~1, pp. 341--381, Jul. 2017.

\bibitem{Kruskal1964}
J.~B. Kruskal, ``Nonmetric multidimensional scaling: A numerical method,''
  \emph{Psychometrika}, vol.~29, no.~2, pp. 115--129, Jun. 1964.

\bibitem{Takane1977}
Y.~Takane, F.~Young, and J.~D. Leeuw, ``Nonmetric individual differences
  multidimensional scaling: An alternating least squares method with optimal
  scaling features,'' \emph{Psychometrika}, vol.~42, no.~1, pp. 7--67, Mar.
  1977.

\bibitem{Parhizkar2013b}
R.~Parhizkar, ``Euclidean distance matrices: Properties, algorithms and
  applications,'' Ph.D. dissertation, School of Computer and Communication
  Sciences, Ecole Polytechnique Federale de Lausanne, 2013.

\bibitem{Khasminskii1976}
R.~Z. Khas'mi, ``A lower bound on the risks of nonparametric estimates of
  densities in the uniform metric,'' \emph{Theory Probab. Appl.}, vol.~23,
  no.~4, pp. 794--798, Dec. 1976.

\bibitem{Plan2014}
Y.~Plan, R.~Vershynin, and E.~Yudovina, ``High-dimensional estimation with
  geometric constraints,'' \emph{Information and Inference: A Journal of the
  IMA}, vol.~6, no.~1, pp. 1--40, Sep. 2017.

\bibitem{Oh2010}
S.~Oh, A.~Montanari, and A.~Karbasi, ``Sensor network localization from local
  connectivity: Performance analysis for the mds-map algorithm,'' in
  \emph{Information Theory (ITW 2010, Cairo), 2010 IEEE Information Theory
  Workshop on}, Jan. 2010, pp. 1--5.

\bibitem{Tsybakov2009}
A.~B. Tsybakov, \emph{Introduction to Nonparametric Estimation}.\hskip 1em plus
  0.5em minus 0.4em\relax Springer, 2009.

\bibitem{Ledoux1991}
M.~Ledoux and M.~Talagrand, \emph{Probability in Banach Spaces: isoperimetry
  and processes}.\hskip 1em plus 0.5em minus 0.4em\relax Springer, 1991.

\bibitem{Seginer2000}
Y.~Seginer, ``The expected norm of random matrices,'' \emph{Combinatorics,
  Probability \& Computing}, vol.~9, no.~2, pp. 149--166, Mar. 2000.

\bibitem{Michael2005}
M.~Michael and U.~Eli, \emph{Probability and Computing: Randomized Algorithms
  and Probabilistic Analysis}.\hskip 1em plus 0.5em minus 0.4em\relax Cambridge
  University Press, 2005.

\bibitem{Roman2012}
R.~Vershynin, ``Introduction to the non-asymptotic analysis of random
  matrices,'' in \emph{Compressed Sensing: Theory and Applications}.\hskip 1em
  plus 0.5em minus 0.4em\relax Cambridge: Cambridge University Press, May 2012,
  pp. 210--268.

\bibitem{Roman2016}
------, \emph{High-Dimensional Probability: An Introduction with Applications
  in Data Science}.\hskip 1em plus 0.5em minus 0.4em\relax Draft, 2016.

\bibitem{Tropp2012}
J.~Tropp, ``User-friendly tail bounds for sums of random matrices,''
  \emph{Found. Comput. Math.}, vol.~12, no.~4, pp. 389--434, Aug. 2012.

\bibitem{Cover2008}
T.~M. Cover and J.~A. Thomas, \emph{Elements of information theory},
  2nd~ed.\hskip 1em plus 0.5em minus 0.4em\relax Wiley-Interscience, 2008.

\bibitem{Negahban2012}
S.~Negahban and M.~J. Wainwright, ``Restricted strong convexity and weighted
  matrix completion: Optimal bounds with noise,'' \emph{J. Mach. Learn. Res.},
  vol.~13, pp. 1665--1697, May 2012.

\bibitem{Hoeffding1963}
W.~Hoeffding, ``Probability inequalities for sums of bounded random
  variables,'' \emph{Journal of the American Statistical Association}, vol.~58,
  no. 301, pp. 13--30, Mar. 1963.

\end{thebibliography}


\end{document}